\documentclass{amsart}

\usepackage{amsmath, amssymb}
\usepackage{float, mathtools, listings, amsthm}
\usepackage[margin=1in]{geometry}
\usepackage{amsrefs}
\usepackage[hidelinks]{hyperref}
\usepackage[T1]{fontenc}
\usepackage[english]{babel}
\usepackage{xcolor}\usepackage{marginnote}

\title{Constructing Approximately Diagonal Quantum Gates}
\date{June 2021}

\author[Griffin]{Colton Griffin}
\author[Cui]{Shawn X. Cui}
\address{Department of Mathematics and Department of Physics and Astronomy\\Purdue University\\
West Lafayette\\
IN 47907\\U.S.A.}
\email{griff254@purdue.edu}

\address{Department of Mathematics and Department of Physics and Astronomy\\Purdue University\\
West Lafayette\\
IN 47907\\U.S.A.}
\email{cui177@purdue.edu}

\newtheorem{thm}{Theorem}[section]

\newtheorem*{thm*}{Theorem}
\newtheorem{lemma}[thm]{Lemma}

\newtheorem{prop}[thm]{Proposition}
\newtheorem*{prop*}{Proposition}
\newtheorem{conjecture}[thm]{Conjecture}
\newtheorem{cor}[thm]{Corollary}

\theoremstyle{definition}
\newtheorem{defn}[thm]{Definition}

\newtheorem{remark}[thm]{Remark}

\newtheorem{example}[thm]{Example}

\newcommand{\SU}{\text{SU}}
\newcommand{\diag}{\text{diag}}
\newcommand{\U}{\text{U}}

\begin{document}

\maketitle

\begin{abstract}
We study a method of producing approximately diagonal 1-qubit gates. For each positive integer, the method provides a sequence of gates that are defined iteratively from a fixed diagonal gate and an arbitrary gate. These sequences are conjectured to converge to diagonal gates doubly exponentially fast and are verified for small integers. We systemically study this conjecture and prove several important partial results. Some techniques are developed to pave the way for a final resolution of the conjecture. The sequences provided here have applications in quantum search algorithms, quantum circuit compilation, generation of leakage-free entangled gates in topological quantum computing, etc.
\end{abstract}

\section{Introduction}\label{section1}
A basic problem in quantum computing is to approximate an arbitrary quantum gate efficiently with a universal gate set. The Solovay-Kitaev theorem provides a general solution to this question. Given a universal gate set in $\SU(d)$, the theorem provides an algorithm to approximate an arbitrary gate of $\SU(d)$  with running time and space complexity both $O(\log^c (1/\epsilon))$ to an accuracy $\epsilon > 0$ \cite{dawson2006solovay}. Here $c \approx 3$ with its explicit value varying depending on the realizations of the theorem. However, this algorithm is usually not optimal and more efficient approximation protocols exist on certain gate sets. Developing optimal approximation protocols is especially critical for systems that have potential experimental implementations. One such example is the Fibonacci anyon circuit, one of the most prominent models for topological quantum computing \cite{freedman2002modular}. In this model, there are algorithms for approximation where the exponent $c$ can be improved to the asymptotically optimal value $c = 1$ \cite{kliuchnikov2014asymptotically, Reichardt2012}.

We focus on the methods used in \cite{Reichardt2012} where a key tool to obtain the optimal $c=1$ is the following proposition. Let $\theta=\frac{\pi}{5}$ and $D(\theta)=\diag(1,e^{i\theta})$. Consider the recursive sequence
\begin{equation}\label{equ:Uk_seq_for_5}
    U_{k+1} = U_{k} D(\theta)U_{k}^{-1} D(\theta)^3 U_{k} D(\theta)^3 U_{k}^{-1}D(\theta)U_{k}.
\end{equation}
It was shown that for any $U_0\in\U(2)$, $|(U_{k+1})_{21}| = |(U_{k})_{21}|^5$ where $(U_k)_{ij}$ denotes the $(i,j)$-entry of $U_k$ \cite{reichardt2005quantum, Reichardt2012}. Hence, if $U_0$ is not diagonal, the sequence in  Equation \ref{equ:Uk_seq_for_5} converges to a diagonal gate\footnote{In fact, the term $D(\theta)^{-7}$ has to be appended to the RHS of Equation \ref{equ:Uk_seq_for_5} in order for the sequence to converge. Otherwise, it would have several convergent subsequences. This will not affect our discussions below though.}. The convergence is double exponentially fast and the space complexity is $O(\log(1/\epsilon))$ to reach the limit within precision $\epsilon > 0$. The above technique is also heavily utilized to design composite pulse sequence for quantum error correction \cite{reichardt2005quantum} and to generate leakage-free entangling 2-qubit gates in the Fibonacci model \cite{carnahan2016systematically, cui2019search}. 

Equation \ref{equ:Uk_seq_for_5} in turn is inspired by a simpler sequence from Grover's quantum search algorithm \cite{Grover2005}. That is, for $\theta = \frac{\pi}{3}$, consider instead the sequence,
\begin{equation}\label{equ:Uk_seq_for_3}
    U_{k+1} = U_{k} D(\theta)U_{k}^{-1} D(\theta) U_{k}.
\end{equation}
Then it is straightforward to check that $|(U_{k+1})_{21}| = |(U_{k})_{21}|^3$. Besides Grover's search algorithm, this sequence is also used in some other quantum algorithms \cite{wocjan2008speedup, wocjan2009quantum}.

The relation between the $(2,1)$-entry (and also the $(1,2)$-entry) of adjacent terms in the above two sequences is intriguing as it is an exact equality. This motivates the question of whether they are special cases of a more general pattern. That is, fix $p=2n+1$ for $n\in\mathbb{N}$, and define $\theta_p\equiv \frac{\pi}{p}=\frac{\pi}{2n+1}$. Is there a recursive sequence $\{U_k^{(n)}\}_{k=0}^{\infty}$, defined similar in form to those in Equations \ref{equ:Uk_seq_for_5} and \ref{equ:Uk_seq_for_3} with $D(\theta)$ replaced by $D(\theta_p)$ such that $|(U_{k+1}^{(n)})_{21}| = |(U_{k}^{(n)})_{21}|^{2n+1}=|(U_{k}^{(n)})_{21}|^{p}$? Such a generalization not only is interesting on its own as a mathematical proposition, but also has applications in topological quantum computing. Recall that the Fibonacci model is described by the Witten-Chern-Simons theory $\SU(2)_3$, in which braiding of anyons naturally gives the diagonal gate $D(\frac{\pi}{5})$\cite{cui2019search}. Hence Equation \ref{equ:Uk_seq_for_5} can be used in this model. The theory $\SU(2)_p$ is also defined for any $p \geq 1$, and for odd $p=2n+1$, braiding of anyons gives the diagonal gate $D(\frac{\pi}{p+2})$. This can be obtained from the $R$-symbols of the theory (Ref. \cite{bonderson2007non}, Sec. 5.4). Therefore, the generalized sequence $\{U_k^{(n+1)}\}_{k=0}^{\infty}$ will be useful in the $\SU(2)_{p}$ anyon model for both topological compilation and generation of entangled gates.

A conjectured formula for the generalized sequence was given for each odd $p$\cite{Reichardt2012} (see also Section \ref{subsec:diagonalization} for an explicit form). For each $p = 2n+1$, the sequence $\{U_k^{(n)}\}$ is defined in a recursive formula similar to those in Equations \ref{equ:Uk_seq_for_5} and \ref{equ:Uk_seq_for_3}. The length of the words in the recursion is $O(n)$. Conjecture \ref{conj:convergence} states that $|(U_{k+1}^{(n)})_{21}| = |(U_k^{(n)})_{21}|^{2n+1}$.

In this paper, we systemically study this conjecture. For each odd $p = 2n+1 > 1$, we analyze the entries of $U_{k+1}^{(n)}$ in terms of those of $U_{k+1}^{(n)}$ and present them in a specially designed form. Explicitly, let $a_k = (U_k^{(n)})_{11}$ and $b_k = (U_k^{(n)})_{21}$ (ignoring the dependence of $a_k$ and $b_k$ on $n$ for now). Then 
\begin{align}
    \frac{b_{k+1}}{b_k} &= \beta_0+|a_k|^2(\beta_1-|b_k|^2(\beta_2+|a_k|^2(\beta_3-|b_k|^2(\ldots)))),
\end{align}
where the coefficients $\beta_0, \cdots, \beta_n$ are expressions involving entries of $D(\theta_{p})$. By using induction, we provide an explicit formula for the $\beta_j\,'$s in terms of roots of unity $\lambda_{j0} = \omega^{(-1)^jj}$ and $\lambda_{j1} = (-1)^{j+1}\omega^{(-1)^{j+1}j}$ for $\omega=e^{i\theta_p/2}$ (Theorem \ref{thm:general_alpha_beta}). In general, these $\beta_j\,'$s are very complicated and hence difficult to evaluate further. However, we are able to compute the values for $\beta_0, \ \beta_1$, and $\beta_n$ for any fixed $n$. Furthermore, we conjecture that the values of all the $\beta_j\,'$s can be expressed as binomial coefficients (Conjecture \ref{conj:values}). The derivation of such values itself is quite non-trivial and involves several technical identities about these binomial coefficients. However, we prove that this secondary conjecture is equivalent to Conjecture \ref{conj:convergence} using Theorem \ref{cor:values_vi}, providing a simpler method of proving the conjecture for a given $n$ by means of showing a set of identities on the coefficients $\beta_j$. A complete verification for the conjectured values would lead to a proof of Conjecture \ref{conj:convergence}. This is left for a future direction. As a concrete application, we prove Conjecture \ref{conj:convergence} for $p = 7$ ($n=3$), with the corresponding sequence given by,
\begin{equation}
    U_{k+1}=U_k D(\theta) U_k^{-1}D(\theta)^{5}U_k D(\theta)^{3}U_k^{-1} D(\theta)^{3}U_k D(\theta)^{5}U_k^{-1} D(\theta)U_k.
\end{equation}

In addition to the results above, we show that two sequences for $p_1$ and $p_2$ respectively can be combined to obtain a sequence for $p_1p_2$ which is different from the one constructed from the conjecture. As a consequence, there exists a sequence for $p= 15 = 3 \cdot 5$ with the desired property but different from the conjectured sequence for $p=15$.

The rest of the paper is organized as follows. In Section \ref{section2}, we provide some backgrounds and define the sequence for each odd $p$ in both the notation of \cite{Reichardt2012} as well as in an alternative form. In Section \ref{section3}, we present the matrix entries in each sequence as a set of solutions to special recursive equations and show how adjacent sequences are related with each other. Sections  \ref{section4} and \ref{section5} are devoted to studying these equations to greater details, including deriving explicit formulas for each $\beta_j$ and evaluating them for $j=0,n$, and 1. Section \ref{section6} provides an alternative approach to obtaining the sequence for a composite integer from those of its prime factors. 

\section{Preliminaries}\label{section2}
We begin with some notation. We use the following construction of a unitary matrix $U_k\in \U(2)$:
\begin{equation}\label{eq:unitary}
U_k=e^{i\varphi_k/2}\begin{pmatrix}
a_k&-\overline{b_k}\\
b_k&\overline{a_k}
\end{pmatrix}
\end{equation}
where $|a_k|^2+|b_k|^2=1$ and $\varphi_k\in [0,2\pi)$, hence $\det U_k=e^{i\varphi_k}$. If $\varphi_k=0$, then $U_k\in\SU(2)$. 
In this notation, the upper left element of $U_k$ is $(U_k)_{11}\equiv e^{i\varphi_k/2}a_k$ and the lower left element of $U_k$ is $(U_k)_{21}\equiv e^{i\varphi_k/2}b_k$. This extra phase will not be important, since throughout most of the paper we will let $\varphi_k=0$ without loss of generality. Hence the upper and lower left elements of $U_k$ will be referred to as $a_k$ and $b_k$ respectively.
We also denote $\chi_j\equiv j\bmod 2$ for brevity since this will be used often throughout the paper. 
\subsection{Diagonalizing Sequences}\label{subsec:diagonalization}
As before, fix $n\in\mathbb{N}$, defining $p\equiv2n+1$. Additionally, fix an input unitary $U_0\in\U(2)$. We wish to construct a sequence $\{U_k^{(n)}\}_{k=0}^{\infty}$ defined recursively from any $U_0$, such that $U_{k+1}$ is expressed as a product of $U_k$, $U_k^{-1}$, and a set of diagonal gates $D_j(\theta)$ to be chosen. The $(n)$ notation here on the sequences denotes which integer $n$ the sequence is defined with respect to. We will also refer to $n$ as the \textbf{order} of the sequence $\{U_{k}^{(n)}\}_{k=0}^{\infty}$. If the order $n$ referred to is clear in context it will be dropped. The objects that will have this notation applied are the elements of our sequences $U_k^{(n)}$ and their sub-elements, such as $(U_k^{(n)})_{11}\equiv e^{i\varphi_k/2}a_k^{(n)}$. Sometimes powers or inverses will be applied, but these will be written without a parenthesis around the superscript.

Let $D_j(\theta_p)$ be diagonal matrices indexed by $j$, with $\theta_p\equiv\frac{\pi}{p}\equiv\frac{\pi}{2n+1}$. These matrices are defined as $D_j(\theta_p)\equiv\diag(\lambda_{j0},\lambda_{j1})$, where
$\lambda_{j0}$ and $\lambda_{j1}$ (we will also refer to these pairs of elements as $\lambda_{jl}$, where $l=0,1$) denote the roots of unity
\begin{equation}\label{eq:lambda}
    \lambda_{j0}=\omega^{(-1)^{j}j},\quad \lambda_{j1}=(-1)^{j+1}\omega^{(-1)^{j+1}j}
\end{equation}
for $\omega=e^{i\theta_p/2}$, and $j=1,\ldots,n$. 

Here we will define the sequences in question.
\begin{defn}
Given some order $n$ and an input matrix $U_0\in\U(2)$, we define the corresponding \textbf{diagonalizing sequence} $\{U_k\}_{k=0}^{\infty} \equiv \{U_k^{(n)}\}_{k=0}^{\infty}$ to be given by the recursive equation $U_{k+1}\equiv Q_nU_k^{(-1)^n}P_n$, where $P_n$ and $Q_n$ are defined recursively as
\begin{equation}
\begin{aligned}
    &P_{j+1}=D_{j+1}(\theta_p)U_k^{(-1)^j}P_j,\\
    &Q_{j+1}=Q_jU_k^{(-1)^j}D_{j+1}(\theta_p),
\end{aligned}
\end{equation}
where $P_0=Q_0=I$.
\end{defn}
Here we are using Reichardt's notation for generating these sequences\cite{Reichardt2012}. For a given $n$, we have that $U_{k+1}$ is expressed as a product of $U_k,U_k^{-1}$, and $D_j(\theta_p)$ for $j=1,\ldots,n$. For example, for order 1 the recursive equation is given by expression
\begin{equation*}
  U_{k+1}=U_k D_1 U_k^{-1}D_1 U_k.
\end{equation*}
For order 2, the equation is given by the expression
\begin{equation*}
    U_{k+1}=U_k D_1 U_k^{-1}D_2 U_k D_2 U_k^{-1} D_1 U_k.
\end{equation*}
In Reichardt's definition of the above sequences he used an alternative set of diagonal matrices, given as $D_j'(\theta_p)\equiv\diag(1,(-1)^{j+1}e^{i\theta_p(-1)^{j+1}j})$, which correspond to $D_j(\theta_p)$ in our formulation.  However, there is no fundamental difference in convergence between these two sequences, as we shall show at the end of this section.

Below we provide an alternative description of the sequences.
\begin{defn}
Given some $U_k\in\U(2)$, let $m$ be any product of any number of matrices $U_{k}$, $U_{k}^{-1}$, and $D_j(\theta)=\diag(\lambda_{j0},\lambda_{j1})$ for any $j$. Define the \textbf{shifting operation} $T$ on $m$ to be a transformation denoted as $Tm$, such that in each entry of $m$ we replace $\theta_p\mapsto \theta_{p+2}$, $\lambda_{jl}\mapsto\lambda_{(j+1),l}$, $a_k\mapsto \overline{a_k}$, and $b_k\mapsto-b_k$.
\end{defn}
An immediate property that we can note is that since we map each instance of every element in the entries of $m$, we can simply distribute $T$ to each element of each matrix in the product $m$ from the definition. Therefore an equivalent mapping for $T$ is $U_k\mapsto U_k^{-1}$ and $D_j(\theta_p)\mapsto D_{j+1}(\theta_{p+2})$.
For a simple example of this shift, in order 1 the diagonalizing sequence is defined by $U_{k+1}=U_kD_1(\theta_p)U^{-1}_kD_1(\theta_p)U_k$. Therefore applying a shift gives ${T}U_{k+1}=U_k^{-1}D_2(\theta_{p+2})U_kD_2(\theta_{p+2})U_k^{-1}$.
This operation allows us to write a compact expression for the above definition of the diagonalizing sequences:
\begin{prop}\label{prop:shift_seq}
For a given order $n$, the diagonalizing sequence $\{U_k^{(n+1)}\}_{k=0}^\infty$ satisfies the recursive relation
\begin{equation}\label{eq:shift}
    U_{k+1}^{(n+1)}=U_kD_1(\theta_{p+2})\,{T}U_{k+1}^{(n)}\, D_1(\theta_{p+2}) U_k,
\end{equation}
where $U_{k+1}^{(n)}$ is treated as a function of $U_k$.
\end{prop}
\begin{proof}
Applying the shift to the definition for $U_{k+1}^{(n)}$, we get $TU_{k+1}^{(n)}=TQ_n\,U_{k}^{(-1)^{n+1}}\,TP_n$ by distribution. The factors $TP_n$ and $TQ_n$ are given by the relations
\[TP_{j+1} = D_{j+2}(\theta_{p+2})U_{k}^{(-1)^{j+1}}TP_{j},\quad TQ_{j+1}=TQ_jU_k^{(-1)^{j+1}}D_{j+2}(\theta_{p+2}).\]
This recursion stops at $j=0$, where $TP_0=TQ_0=I$. Now consider the products $\tilde P_{n+1}$ and $\tilde Q_{n+1}$ defined by $\tilde P_{j+1}=D_{j+1}(\theta_{p+2})U_k^{(-1)^{j}}\tilde P_j$ and $\tilde Q_{j+1}=\tilde Q_jU_k^{(-1)^{j}}D_{j+1}(\theta_{p+2})$ for $j=0,\ldots,n$. The tilde on the terms $\tilde P_j$ and $\tilde Q_j$ denotes using the angle $\theta_{p+2}$ instead of $\theta_{p}$, which was the angle $P_n$ originally used before shifting by $T$.
Note that by definition $U_{k+1}^{(n+1)}=\tilde Q_{n+1}U_k^{(-1)^{n+1}}\tilde P_{n+1}$. Left multiplication by $(TP_{j})^{-1}$ on $\tilde P_{j+1}$ gives us
\begin{align*}
    (TP_{j})^{-1}\tilde P_{j+1} &= \left((TP_{j-1})^{-1}U_k^{(-1)^{j+1}}D_{j+1}^{-1}(\theta_{p+2})\right)D_{j+1}(\theta_{p+2})U_k^{(-1)^j}\tilde P_j\\
    &=(TP_{j-1})^{-1}\tilde P_{j}.
\end{align*}
By induction, we get a chain of equalities $(TP_{n})^{-1}\tilde P_{n+1}=\ldots=(TP_{0})^{-1}\tilde P_{1}=D_1(\theta_{p+2})U_k$. Therefore $\tilde P_{n+1}=TP_{n}\,D_1(\theta_{p+2})U_k$. The same process applies for $\tilde Q_{n+1}$, and we conclude that $\tilde Q_{n+1}=U_kD_1(\theta_{p+2})\,TQ_{n}$. Therefore
\begin{align*}
    U_{k+1}^{(n+1)}&=\tilde Q_{n+1}U_k^{(-1)^{n+1}}\tilde P_{n+1}\\
    &=U_kD_1(\theta_{p+2})\,TQ_{n}U_k^{(-1)^{n+1}}TP_{n}\,D_1(\theta_{p+2})U_k\\
    &=U_kD_1(\theta_{p+2})TU_{k+1}^{(n)}D_1(\theta_{p+2})U_k.
\end{align*}
\end{proof}
For the rest of this paper, we will use this alternative form for the diagonalizing sequences instead of the original definition.
The conjectured property of these sequences is the following:
\begin{conjecture}\label{conj:convergence}
Given any $U_0\in\U(2)$ and a fixed order $n$, the diagonalizing sequence $\{U_{k}\}_{k=0}^{\infty}$ has the property that
\begin{equation}
    |(U_{k+1})_{21}| \equiv |b_{k+1}|= |b_{k}|^{2n+1}.
\end{equation}
\end{conjecture}
If $|b_0|=1$, then $|b_{k+1}|=|b_k|^{2n+1}=\ldots=|b_0|^{(2n+1)^{k+1}}=1$. Since $|b_0|\in[0,1]$, this conjecture implies that the sequence will converge to a diagonal gate unless the input gate $U_0$ is skew-diagonal.
Before we begin the analysis, we will prove two more properties of diagonalizing sequences.
\begin{prop}\label{prop:det}
For a given order $n$ and a diagonalizing sequence $\{U_k\}_{k=0}^\infty$, we have that $\det (U_k) = \det (U_0)$. Therefore $\det (U_k)=1$ if  $\det (U_0)=1$.
\end{prop}
\begin{proof}
Note that $\det(D_j(\theta_{p}))=\lambda_{j0}\lambda_{j1} = (-1)^{j+1}\omega^{(-1)^{j+1}j}\omega^{(-1)^jj}=(-1)^{j+1}$, so $\det(D_j(\theta_p))^2=1$. The determinants of $P_{j+1}$ and $Q_{j+1}$ are given by
\begin{equation*}
\begin{aligned}
    &\det (P_{j+1})=\det(D_{j+1}(\theta_p))\det(U_k)^{(-1)^j}\det(P_j),\\
    &\det (Q_{j+1})=\det(Q_j)\det(U_k)^{(-1)^j}\det(D_{j+1}(\theta_p)),
\end{aligned}
\end{equation*}
and so $\det(Q_{j+1})/\det(Q_j)=\det(P_{j+1})/\det(P_j)$. Since $\det(Q_1)=\det(P_1)=\det(U_k)\det(D_1(\theta_p))$, we have $\det(Q_2)=\det(P_2)$ and by induction $\det(Q_n)=\det(P_n)$. Therefore $\det(U_{k+1})=\det(Q_n)^2\det(U_k)^{(-1)^n}$. We claim that $\det(Q_j)^2=1$ if $j$ is even and $\det(Q_j)^2=\det(U_k)^2$ if $j$ is odd. We have that $\det(Q_1)^2=\det(U_k)^2$, and by induction we assume that the result holds for $j$. For odd $j$ we have 
$\det(Q_{j+1})^2 = \det(U_k)^{-2}\det(Q_j)^2=1$, and for even $j$ we have $\det(Q_{j+1})^2 = \det(U_k)^{2}\det(Q_j)^2=\det(U_k)^2$, which implies the statement for $j+1$. By induction we conclude that the result holds for $j=n$. For odd $n$, we get $\det(U_{k+1}) = \det(Q_n)^2\det(U_k)^{-1} = \det(U_k)^2\det(U_k)^{-1}=\det(U_k)$, and for even $n$ we get $\det(U_{k+1}) = \det(Q_n)^2\det(U_k) = 1\cdot\det(U_k)=\det(U_k)$. This shows the result, as $\det(U_{k+1})=\det(U_k)=\ldots=\det(U_0)$.
\end{proof}
\begin{prop}\label{prop:phase}
Given a diagonalizing sequence $\{U_k\}_{k=0}^\infty$ of order $n$ and input $U_0$, consider the sequence $\{\tilde U_k\}_{k=0}^\infty$ defined by the recursive relation $\tilde U_{k} = e^{i\gamma_k}U_{k}$ for $\gamma_k\in\mathbb{R}$. If $\{U_k\}$ has the property from Conjecture \ref{conj:convergence}, then so does $\{\tilde U_k\}$.
\end{prop}
\begin{proof}
For each $k$ in the sequence $\{\tilde U_k\}$, the lower left element $(\tilde U_k)_{21}$ will have the same norm as $(U_k)_{21}$,
and so $|(\tilde U_{k+1})_{21}|=|(U_{k+1})_{21}|=|(U_{k})_{21}|^{2n+1}=|(\tilde U_{k})_{21}|^{2n+1}$.
\end{proof}
Without loss of generality, we will set $\det U_0=1$ for convenience, as for any input matrix $U_0$ with $\det U_k=e^{i\varphi_0}$ we satisfy the property from Conjecture \ref{conj:convergence} if and only if we have the same property for the input matrix $\tilde U_0=e^{-i\varphi_0/2}U_0$ with determinant $\det \tilde U_0 = 1$.
\begin{remark}
We can pull the $\lambda_{j0}$ factor out for each matrix $D_j$ and obtain the diagonal matrices $D_j'=\diag(1,(-1)^{j+1}e^{i\theta_p(-1)^{j+1}j})$, the same definition as given by Reichardt in Equation 6\cite{Reichardt2012}. The sequence $\{U_k'\}$ defined in the same way as $\{U_k\}$ but replacing $D_j$ with $D_j'$ satisfies the property from Conjecture \ref{conj:convergence} if and only if $\{U_k\}$ has the property by Proposition \ref{prop:phase}. This is because $\tilde U_k$ and $U_k$ differ by a phase $U_k=(\prod_{j=1}^n\lambda_{j0}^2)\tilde U_k$.
That said, our choice for $D_j$ is due to convenient properties like Proposition \ref{prop:det}, which this sequence does not have; when we refer to $\lambda_{j0}$ and $\lambda_{j1}$, we mean the definitions in equation \ref{eq:lambda} unless stated otherwise.
\end{remark}

\section{Recursively Constructing Sequences of Arbitrary Order}\label{section3}
In this section we show our elements come in the form $(U_{k+1})_{11}=a_{k+1}=a_k\mathcal{A}_k$ and $(U_{k+1})_{21}=b_{k+1}=b_k\mathcal{B}_k$, where $\mathcal{A}_k$ and $\mathcal{B}_k$ are polynomial functions in terms of $\lambda_{jl}$ and $b_k,a_k$. This ansatz will allow us to examine the Conjecture \ref{conj:convergence} in a way that is independent of $U_0$.
\begin{prop}\label{prop:ansatz}
For any order $n$, the first column of $U_{k+1}$ comes in the form $(U_{k+1})_{11}=a_{k+1}=a_k\mathcal{A}_k$ and $(U_{k+1})_{21}=b_{k+1}=b_k\mathcal{B}_k$, where $\mathcal{A}_k$ and $\mathcal{B}_k$ are polynomials of the form
\begin{equation}\label{eq:ansatz}
\begin{aligned}
    \mathcal{B}_k&=\beta_0+|a_k|^2\left(\sum_{j=1}^{\lfloor{n/2}\rfloor}(-|b_k|^2|a_k|^2)^{j-1}(\beta_{2j-1}-|b_k|^2\beta_{2j})\right),\\
    \mathcal{A}_k&=\alpha_0-|b_k|^2\left(\sum_{j=1}^{\lfloor{n/2}\rfloor}(-|b_k|^2|a_k|^2)^{j-1}(\alpha_{2j-1}+|a_k|^2\alpha_{2j})\right),
\end{aligned}
\end{equation}
where $\beta_j$ and $\alpha_j$ are defined to be 0 if $j>n$. For each $j=1,\ldots,n$, the coefficients $\beta_j$ and $\alpha_j$ are polynomial expressions of $\lambda_{i0}$ and $\lambda_{i1}$ for $i=1,\ldots,n$.
\end{prop}
\begin{proof}
We are going to use induction on $n$ to prove this ansatz is correct, with the base case $n=1$. Using the above ansatz we can express $U_{k+1}$ and the shifted $TU_{k+1}$ to be of the form
\begin{equation}
    U_{k+1}=\begin{psmallmatrix}
        {a_{k+1}}&-\overline{b_{k+1}}\\
        {b_{k+1}}&\overline{a_{k+1}}
    \end{psmallmatrix}\equiv\begin{psmallmatrix}
        {a_k}{\mathcal{A}}_k&-\overline{b_k}\overline{{\mathcal{B}}_k}\\
        {b_k}{\mathcal{B}}_k&\overline{a_k}\overline{{\mathcal{A}}_k}
    \end{psmallmatrix}\implies {T}U_{k+1}=
    \begin{psmallmatrix}
        \overline{a_k}{T}{\mathcal{A}}_k&\overline{b_k}\overline{{T}{\mathcal{B}}_k}\\
        -{b_k}{T}{\mathcal{B}}_k&{a_k}\overline{{T}{\mathcal{A}}_k}
    \end{psmallmatrix}.
\end{equation}
Now we insert this into our sequence $U_{k+1}^{(n+1)}$. The calculation below will be useful for showing how case $n$ implies case $n+1$ as well as for examining the base case. Assume that $\mathcal{A}_k^{(n)}$ and $\mathcal{B}_k^{(n)}$ take the form of the ansatz. Then expanding out the definition of the diagonalizing sequence in Proposition \ref{prop:shift_seq} yields
\begin{gather*}
    U_{k+1}^{(n+1)}=U_kD_1\, {T}U_{k+1}^{(n)}\, D_1 U_k\\
    =
    \begin{psmallmatrix}
        a_k&-\overline{b_k}\\
        b_k&\overline{a_k}
    \end{psmallmatrix}
    \begin{psmallmatrix}
        \lambda_{10}&0\\
        0&\lambda_{11}
    \end{psmallmatrix}
    \begin{psmallmatrix}
        \overline{a_k}{T}{\mathcal{A}}_k&\overline{b_k}\overline{{T}{\mathcal{B}}_k}\\
        -{b_k}{T}{\mathcal{B}}_k&{a_k}\overline{{T}{\mathcal{A}}_k}
    \end{psmallmatrix}
    \begin{psmallmatrix}
        \lambda_{10}&0\\
        0&\lambda_{11}
    \end{psmallmatrix}
    \begin{psmallmatrix}
        a_k&-\overline{b_k}\\
        b_k&\overline{a_k}
    \end{psmallmatrix}\\
    =\begin{psmallmatrix}
        \lambda_{10}a_k&-\lambda_{11}\overline{b_k}\\
        \lambda_{10}b_k&\lambda_{11}\overline{a_k}
    \end{psmallmatrix}
    \begin{psmallmatrix}
        \overline{a_k}{T}{\mathcal{A}}_k&\overline{b_k}\overline{{T}{\mathcal{B}}_k}\\
        -{b_k}{T}{\mathcal{B}}_k&{a_k}\overline{{T}{\mathcal{A}}_k}
    \end{psmallmatrix}
    \begin{psmallmatrix}
        \lambda_{10}a_k&-\lambda_{10}\overline{b_k}\\
        \lambda_{11} b_k&\lambda_{11}\overline{a_k}
    \end{psmallmatrix}\\
    =
    \begin{psmallmatrix}
        \lambda_{10}a_k&-\lambda_{11}\overline{b_k}\\
        \lambda_{10}b_k&\lambda_{11}\overline{a_k}
    \end{psmallmatrix}
    \begin{psmallmatrix}
        \lambda_{10}|a_k|^2{T}{\mathcal{A}}_k+\lambda_{11}|b_k|^2\overline{{T}{\mathcal{B}}_k}&-\overline{a_kb_k}(\lambda_{10}{T}{\mathcal{A}}_k-\lambda_{11}\overline{{T}{\mathcal{B}}_k})\\
        a_kb_k(\lambda_{11}\overline{{T}{\mathcal{A}}_k}-\lambda_{10}{{T}{\mathcal{B}}_k})&\lambda_{11}|a_k|^2\overline{{T}{\mathcal{A}}_k}+\lambda_{10}|b_k|^2{{T}{\mathcal{B}}_k}
    \end{psmallmatrix}\\
    =\begin{psmallmatrix}
        (U_{k+1}^{(n+1)})_{11}&(U_{k+1}^{(n+1)})_{12}\\
        (U_{k+1}^{(n+1)})_{21}&(U_{k+1}^{(n+1)})_{22}
    \end{psmallmatrix},
\end{gather*}
where each element of the matrix is given by the forms
\begin{align*}
    (U_{k+1}^{(n+1)})_{11} &= a_k(\lambda_{10}^2|a_k|^2{T}{\mathcal{A}}_k+\lambda_{10}\lambda_{11}|b_k|^2\overline{{T}{\mathcal{B}}_k} - |b_k|^2(\lambda_{11}^2\overline{{T}{\mathcal{A}}_k}-\lambda_{10}\lambda_{11}{{T}{\mathcal{B}}_k}))\\
    &=a_k(\lambda_{10}^2{T}{\mathcal{A}}_k - |b_k|^2(\lambda_{10}^2{T}{\mathcal{A}}_k + \lambda_{11}^2\overline{{T}{\mathcal{A}}_k}-\lambda_{10}\lambda_{11}{{T}{\mathcal{B}}_k}-\lambda_{10}\lambda_{11}\overline{{T}{\mathcal{B}}_k})),
\end{align*}
\begin{align*}
    (U_{k+1}^{(n+1)})_{21}&=b_k(\lambda_{10}^2|a_k|^2{T}{\mathcal{A}}_k+\lambda_{10}\lambda_{11}|b_k|^2\overline{{T}{\mathcal{B}}_k} + |a_k|^2(\lambda_{11}^2\overline{{T}{\mathcal{A}}_k}-\lambda_{10}\lambda_{11}{{T}{\mathcal{B}}_k}))\\
    &=b_k(\lambda_{10}\lambda_{11}\overline{{T}{\mathcal{B}}_k} + |a_k|^2(\lambda_{10}^2{T}{\mathcal{A}}_k + \lambda_{11}^2\overline{{T}{\mathcal{A}}_k}-\lambda_{10}\lambda_{11}{{T}{\mathcal{B}}_k}-\lambda_{10}\lambda_{11}\overline{{T}{\mathcal{B}}_k})),
\end{align*}
\begin{align*}
    (U_{k+1}^{(n+1)})_{12}&=-\overline{b_k}(\lambda_{11}^2|a_k|^2\overline{{T}{\mathcal{A}}_k}+\lambda_{10}\lambda_{11}|b_k|^2{{T}{\mathcal{B}}_k} + |a_k|^2(\lambda_{10}^2{{T}{\mathcal{A}}_k}-\lambda_{10}\lambda_{11}\overline{{T}{\mathcal{B}}_k}))\\
    &=-\overline{b_k}({\lambda_{10}\lambda_{11}{T}{\mathcal{B}}_k} + |a_k|^2(\lambda_{10}^2{{T}{\mathcal{A}}_k}+\lambda_{11}^2\overline{{T}{\mathcal{A}}_k}-\lambda_{10}\lambda_{11}{{T}{\mathcal{B}}_k}-\lambda_{10}\lambda_{11}\overline{{T}{\mathcal{B}}_k}),
\end{align*}
\begin{align*}
    (U_{k+1}^{(n+1)})_{22}&=\overline{a_k}(\lambda_{11}^2|a_k|^2\overline{{T}{\mathcal{A}}_k}+\lambda_{10}\lambda_{11}|b_k|^2{{T}{\mathcal{B}}_k} - |b_k|^2(\lambda_{10}^2{{T}{\mathcal{A}}_k}-\lambda_{10}\lambda_{11}\overline{{T}{\mathcal{B}}_k}))\\
    &=\overline{a_k}(\lambda_{11}^2\overline{{T}{\mathcal{A}}_k} - |b_k|^2(\lambda_{10}^2{{T}{\mathcal{A}}_k}+\lambda_{11}^2\overline{{T}{\mathcal{A}}_k}-\lambda_{10}\lambda_{11}{{T}{\mathcal{B}}_k}-\lambda_{10}\lambda_{11}\overline{{T}{\mathcal{B}}_k})).
\end{align*}
In each secondary line, we have reorganized the terms in $(U_{k+1}^{(n+1)})_{ij}$ by using the identity $|a_k|^2+|b_k|^2=1$. Note that the shifted $T\mathcal{A}_k$ and $T\mathcal{B}_k$ are still of the form of the ansatz since $T|a_k|^2 = T(a_k\overline{a_k})=T(\overline{a_k}{a_k})=|a_k|^2$ and $T|b_k|^2=T((-b_k)(-\overline{b_k}))=|b_k|^2$. In $(U_{k+1}^{(n+1)})_{11}$, the terms $T\mathcal{B}_k+\overline{T\mathcal{B}_k}$ preserve the form of $\mathcal{A}_k^{(n+1)}$ since they both have a $|a_k|^2$ factor in front. For $T\mathcal{A}_k$ and $\overline{T\mathcal{A}_k}$, we substitute $|a_k|^2=1-|b_k|^2$ and $|b_k|^2=1-|a_k|^2$ to obtain the following equivalent expressions for the ansatz:
\begin{equation*}
\begin{aligned}
    \mathcal{B}_k
    &=\beta_0+\sum_{j=1}^{\lfloor{n/2}\rfloor}(-|b_k|^2|a_k|^2)^{j-1}\beta_{2j-1}\\
    &-|b_k|^2\sum_{j=1}^{\lfloor{n/2}\rfloor}(-|b_k|^2|a_k|^2)^{j-1}(\beta_{2j-1}+|a_k|^2\beta_{2j}),\\
    \mathcal{A}_k
    &=\alpha_0-\sum_{j=1}^{\lfloor{n/2}\rfloor}(-|b_k|^2|a_k|^2)^{j-1}\alpha_{2j-1}\\
    &+|a_k|^2\sum_{j=1}^{\lfloor{n/2}\rfloor}(-|b_k|^2|a_k|^2)^{j-1}(\alpha_{2j-1}-|b_k|^2\alpha_{2j}).
\end{aligned}
\end{equation*}
Therefore we write $\mathcal{B}_k$ and $\mathcal{A}_k$ alternatively in the form
\begin{equation}\label{eq:beta_alt}
    \mathcal{B}_k=\sum_{j=0}^{\lfloor{n/2}\rfloor}(-|b_k|^2|a_k|^2)^{j}(\beta_{2j}+\beta_{2j+1}-|b_k|^2\beta_{2j+1}),
\end{equation}
\begin{equation}\label{eq:alpha_alt}
    \mathcal{A}_k=\sum_{j=0}^{\lfloor{n/2}\rfloor}(-|b_k|^2|a_k|^2)^{j}(\alpha_{2j}-\alpha_{2j+1}+|a_k|^2\alpha_{2j+1}).
\end{equation}
Plugging the alternative forms in, we can see that $T\mathcal{A}_k$ and $\overline{T\mathcal{A}_k}$ also preserves the form for $\mathcal{A}_k^{(n+1)}$. We can apply the same approach for the other entries in the matrix to see that $\mathcal{B}_k^{(n+1)}$ also takes the form of the ansatz.
    
Now we consider the base case $n=1$. Substituting $T\mathcal{A}_k\to 1$ and $T\mathcal{B}_k\to 1$, this corresponds to the original sequence examined by Grover, $U_{k+1}=U_kD_1U_k^{-1}D_1U_k$:
\begin{align*}
    (U_{k+1})_{11}&=a_k(\lambda_{10}^2 - |b_k|^2(\lambda_{10}^2 + \lambda_{11}^2-\lambda_{10}\lambda_{11}-\lambda_{10}\lambda_{11})),\\
    (U_{k+1})_{21}&=b_k(\lambda_{10}\lambda_{11} + |a_k|^2(\lambda_{10}^2 + \lambda_{11}^2-\lambda_{10}\lambda_{11}-\lambda_{10}\lambda_{11})).
\end{align*}
In the form of the ansatz above, we can see that $\alpha_0=\lambda_{10}^2$, $\beta_0=\lambda_{10}\lambda_{11}$, and $\alpha_1=\beta_1=(\lambda_{10}-\lambda_{11})^2$. Since this form is true for $n=1$, it must continue to hold for any $n$ by induction.
\end{proof}
In the statement of the above proposition, we claimed that $\alpha_j$ and $\beta_j$ are polynomials in the $\lambda_{jl}$'s. We derive their exact forms later in this paper, in Theorem \ref{thm:general_alpha_beta}. For reference, the reader may look at Appendix B for the explicit forms for $\alpha_j^{(n)}$ and $\beta_j^{(n)}$, for $n=1,2$, and 3. These forms were obtained by directly expanding the sequences and writing the sequences in the form of the ansatz.


\subsection{Constructing a Recursive System of Equations}
From Proposition \ref{prop:ansatz}, we can conclude that for order $n+1$, we have
\begin{equation}
\begin{aligned}
    \frac{a_{k+1}}{a_k}=\mathcal{A}_{k}^{(n+1)}=\lambda_{10}^2{T}{\mathcal{A}}_k^{(n)} &- |b_k|^2\Big(\lambda_{10}^2{T}{\mathcal{A}}_k^{(n)} + \lambda_{11}^2\overline{{T}{\mathcal{A}}_k^{(n)}}\\
    &-\lambda_{10}\lambda_{11}{{T}{\mathcal{B}}_k^{(n)}}-\lambda_{10}\lambda_{11}\overline{{T}{\mathcal{B}}_k^{(n)}}\Big),
\end{aligned}
\end{equation}
\begin{equation}
\begin{aligned}
    \frac{b_{k+1}}{b_k}=\mathcal{B}_{k}^{(n+1)}=\lambda_{10}\lambda_{11}\overline{{T}{\mathcal{B}}_k^{(n)}} &+ |a_k|^2\Big(\lambda_{10}^2{T}{\mathcal{A}}_k^{(n)} + \lambda_{11}^2\overline{{T}{\mathcal{A}}_k^{(n)}}\\
    &-\lambda_{10}\lambda_{11}{{T}{\mathcal{B}}_k^{(n)}}-\lambda_{10}\lambda_{11}\overline{{T}{\mathcal{B}}_k^{(n)}}\Big).
\end{aligned}
\end{equation}
We can assume $\mathcal{B}_k$ is real for any $k$ and $n$ by induction on $k$, since $\lambda_{10}\lambda_{11}=1$ and each complex term is summed with its conjugate, along with the fact that $\mathcal{B}_k$ is real for $n=1$. Including this assumption allows us to create the following system of equations by aligning terms together and using equations \ref{eq:beta_alt} and \ref{eq:alpha_alt} to preserve the ansatz form:
\begin{equation}\label{eq:alpha}
    \alpha_j^{(n+1)}=\lambda_{10}^2{T}{\alpha}_{j}^{(n)}+\lambda_{10}^2({T}{\alpha}_{j-1}^{(n)}-\chi_j{T}{\alpha}_{j}^{(n)})+\lambda_{11}^2\left(\overline{{T}{\alpha}_{j-1}^{(n)}}-\chi_j\overline{{T}{\alpha}_{j}^{(n)}}\right)-2\lambda_{10}\lambda_{11}{T}{\beta}_{j-1}^{(n)},
\end{equation}
\begin{equation}\label{eq:beta}
    \beta_j^{(n+1)}=\lambda_{10}\lambda_{11}{T}{\beta}_j^{(n)}+\lambda_{10}^2{T}{\alpha}_{j-1}^{(n)}+\lambda_{11}^2\overline{{T}{\alpha}_{j-1}^{(n)}}-2\lambda_{10}\lambda_{11}\left({T}{\beta}_{j-1}^{(n)}+\chi_j{T}{\beta}_{j}^{(n)}\right).
\end{equation}
We denote $\chi_j\equiv j\bmod 2$ as before, so when we align terms from $\mathcal{A}_j^{(n)}$, $\mathcal{A}_{j-1}^{(n)}$, $\mathcal{B}_j^{(n)}$, and $\mathcal{B}_{j-1}^{(n)}$, the $\chi_j$ represents the alternating additional term in equations \ref{eq:beta_alt} and \ref{eq:alpha_alt}.
This approach reduces the study of the sequence $\{U_{k}^{(n+1)}\}$ into solving the above two equations, and by solving them we find closed-form expressions for every $\alpha_j$ and $\beta_j$.

If $j$ is odd, then we get that this simplifies to
\begin{equation}
\begin{aligned}
    \alpha_j^{(n+1)}&=\lambda_{10}^2{T}{\alpha}_{j-1}^{(n)}+\lambda_{11}^2\overline{{T}{\alpha}_{j-1}^{(n)}}-2\lambda_{10}\lambda_{11}{T}{\beta}_{j-1}^{(n)}-\lambda_{11}^2\overline{{T}{\alpha}_{j}^{(n)}}\\
    \beta_j^{(n+1)}&=\lambda_{10}^2{T}{\alpha}_{j-1}^{(n)}+\lambda_{11}^2\overline{{T}{\alpha}_{j-1}^{(n)}}-2\lambda_{10}\lambda_{11}{T}{\beta}_{j-1}^{(n)}-\lambda_{10}\lambda_{11}{T}{\beta}_j^{(n)}.
\end{aligned}    
\end{equation}
For even $j$, it becomes
\begin{equation}
\begin{aligned}
    \alpha_j^{(n+1)}&=\lambda_{10}^2{T}{\alpha}_{j-1}^{(n)}+\lambda_{11}^2\overline{{T}{\alpha}_{j-1}^{(n)}}-2\lambda_{10}\lambda_{11}{T}{\beta}_{j-1}^{(n)}+\lambda_{10}^2{{T}{\alpha}_{j}^{(n)}}\\
    \beta_j^{(n+1)}&=\lambda_{10}^2{T}{\alpha}_{j-1}^{(n)}+\lambda_{11}^2\overline{{T}{\alpha}_{j-1}^{(n)}}-2\lambda_{10}\lambda_{11}{T}{\beta}_{j-1}^{(n)}+\lambda_{10}\lambda_{11}{T}{\beta}_j^{(n)}.
\end{aligned}    
\end{equation}
The goal of section 5 will be to examine these equations in detail, ultimately deriving the general solution to both in Theorem \ref{thm:general_alpha_beta}.

\begin{remark}
Although our assumptions allow us to assume $U_k^{(n)}\in\SU(2)$ for any $k$, this is not necessary in order to find a convergent sequence. However, it is highly beneficial to do so since the recursive equations below are much easier to solve given our assumptions on $\lambda_{jl}$. Instead, we could provide some general form using functions $\mathcal{A}_k,\mathcal{B}_k,\mathcal{C}_k,\mathcal{D}_k$:
\[U_{k+1}^{(n)}=\begin{psmallmatrix}
    a_k\mathcal{A}_k&\overline{b_k}\mathcal{C}_k\\
    b_k\mathcal{B}_k&\overline{a_k}\mathcal{D}_k
\end{psmallmatrix}.\]
and we would have a set of 4 recursive equations to solve simultaneously. This approach would work when analyzing related recursive sequences of matrices.
\end{remark}
\section{Deriving Necessary Values for Convergence}\label{section4}
In order for a diagonalizing sequence to have the property in Conjecture \ref{conj:convergence}, the coefficients $\beta_j$ must take specific numerical values. These values simplify the expression for $\mathcal{B}_k$ in equation \ref{eq:ansatz} to $\mathcal{B}_k=\beta_0|b_k|^{2n}$, which implies the property since $|b_{k+1}|=|\mathcal{B}_kb_k|=|b_k|^{2n+1}$ ($\beta_0=(-1)^{\frac{n\left(n+3\right)}{2}}$, which will be shown in Theorem \ref{eq:j0}). In this section we derive what those values are (Theorem \ref{cor:values_vi}), and first we show how the expression simplifies via the relation $|a_k|^2=1-|b_k|^2$.
\begin{prop}\label{prop:necessary}
Suppose that for some order $n$ and $\mathcal{B}_k$ given by the ansatz in Equation \ref{eq:ansatz}, we substitute $|a_k|^2=1-|b_k|^2$ in every instance of $|a_k|^2$. Then $\mathcal{B}_k$ takes the form
\[\mathcal{B}_k=\beta_0[1+(1-|b_k|^2)(A_1+A_2|b_k|^2+\ldots+A_{n}|b_{k}|^{2n-2})],\]
where each $A_i$ is some integral linear combination of the coefficients $\beta_j/\beta_0$.
Given this form, $\mathcal{B}_k = \beta_0|b_k|^{2n}$ for any $b_k\neq 0$ if and only if $A_i=-1$ for each $i$.
\end{prop}
\begin{proof}
Suppose that $A_i=-1$ for all $i$, then
\begin{align*}
    \mathcal{B}_k&=\beta_0[1+(1-|b_k|^{2})(-1-|b_k|^2-\ldots-|b_k|^{2n-2})]\\
    &=\beta_0[1-(1-|b_k|^{2})(1+|b_k|^2+\ldots+|b_k|^{2n-2})]\\
    &= \beta_0[1-(1-|b_k|^{2n})]=\beta_0|b_k|^{2n} .
\end{align*}
Supposing that $\mathcal{B}_k=\beta_0|b_k|^{2n}$, we work backward through these equations to the first line.
Therefore we conclude that $-1-|b_k|^2-\ldots-|b_k|^{2n-2}=A_1+A_2|b_k|^2+\ldots+A_n|b_k|^{2n-2}$, and this implies that $A_i=-1$ for all $i$ since this identity holds for $|b_k|\in(0,1]$.
\end{proof}
Note that in the above proof, if $b_k=0$, then only $A_1=-1$ is required to guarantee that $\mathcal{B}_k = \beta_0|b_k|^{2n}=0$. This is not a problem for the stated conjecture since $b_k=0$ implies that $U_k$ is a diagonal matrix, so $U_{k+1}$ is also diagonal and $|b_{k+1}|=|b_k|^{2n+1} = 0$.

Now we want to derive what the coefficients of the combinations $A_i$ are, and to do so we will make the exact substitution suggested in the proposition.
Before stating the theorem, here is a low-order example to demonstrate how these combinations arise.
\begin{example}\label{ex:n4matrix}
For order 4 we apply the relation $|a_k|^2+|b_k|^2=1$ to obtain $\mathcal{B}_k$ in terms of the coefficients and powers of $|b_k|^2$:
\begin{align*}
    \mathcal{B}_k&=\beta_0+|a_k|^2(\beta_1-|b_k|^2(\beta_2+|a_k|^2(\beta_3-|b_k|^2\beta_4)))\\
    &=\beta_0+(1-|b_k|^2)(\beta_1+\beta_2(-|b_k|^2)+\beta_3(-|b_k|^2+|b_k|^4)+\beta_4(|b_k|^4-|b_k|^6)).
\end{align*}
Defining $v_j \equiv \beta_j/\beta_0$, $\mathbf{v}\equiv(v_1,v_2,v_3,v_4)$, and grouping the coefficients in front of each power of $|b_k|^2$, we get
\[\mathcal{B}_k = \beta_0\left[1+(1-|b_k|^2)(v_1 + |b_k|^2(-v_2-v_3)+|b_k|^4(v_3+v_4)+|b_k|^6(-v_4))\right].\]
Suppose that these constants $v_j$ for $j=1,2,3,4$ happened to satisfy the linear system
\begin{equation}
\begin{aligned}
    v_1&=-1,\\
    -v_2-v_3&=-1,\\
    v_3+v_4&=-1,\\
    v_4&=-1,
\end{aligned}\quad \iff\quad 
    \begin{bmatrix}
1&0&0&0\\0&-1&-1&0\\0&0&1&1\\0&0&0&-1
\end{bmatrix}
\begin{bmatrix}
v_1\\v_2\\v_3\\v_4
\end{bmatrix}=
\begin{bmatrix}
-1\\-1\\-1\\-1
\end{bmatrix}.
\end{equation}
Then we would be able to simplify the above to be
\[\mathcal{B}_k = \beta_0\left[1+(1-|b_k|^2)(-1 - |b_k|^2-|b_k|^4-|b_k|^6)\right]=\beta_0\left[1+(|b_k|^8-1)\right]=\beta_0|b_k|^8.\]
Denoting this matrix as $M_4$, we note that this is an upper triangular matrix with determinant 1, so it is invertible.
Denote each row of $M_4$ as $\mathbf{r}_i^T$; we identify the rows as the tuple of coefficients in the combination $A_i$, such that $\mathbf{r}_i^T\mathbf{v}=A_i=-1$.
As it turns out, if $\mathcal{B}_k$ did happen to simplify down to $\beta_0|b_k|^8$, then by the invertible nature of the system we must have that $\mathbf{v}$ is given by $\mathbf{v}=M_4^{-1}(-\mathbf{e})=(-1,3,-2,1)$, where $\mathbf{e}=(1,1,1,1)$. This implies a direct link between the values of the coefficients $\beta_j$ and the property from Conjecture \ref{conj:convergence}.
\end{example}
Now we will show what the general form of $M_n$ for any order $n$ is by construction.
\begin{defn}\label{def:matrix}
For a given $n\in\mathbb{N}$, define the matrix $M_n\in\mathbb{R}^{n\times n}$ such that its entries are given as $[M_n]_{ii}=(-1)^{i+1}$ and
\begin{gather}
[{M}_n]_{(2m-1-j),(2m-1)}=(-1)^{j}\binom{m-1}{j},\\
[{M}_n]_{(2m-j),(2m)}=(-1)^{j+1}\binom{m-1}{j},
\end{gather}
for $j,m\geq 1$. All other elements of the matrix are 0.
\end{defn}
Here is an explicit calculation of $M_n$ for a higher order, also including its inverse.
\begin{example}
One interesting observation is that ${M}_{m}$ contains ${M}_{n}$ as a nested sub-block for $m>n$. We will solve the matrix system by finding the inverse matrix $M_n^{-1}$, and it helps to have an example to see the patterns in the entries. Explicitly calculating $M_{10}$ yields the matrix
\begin{equation}
    {M}_{10}=\begin{bmatrix}
1&0&0&0&0&0&0&0&0&0\\
0&-1&-1&0&0&0&0&0&0&0\\
0&0&1&1&1&0&0&0&0&0\\
0&0&0&-1&-2&-1&-1&0&0&0\\
0&0&0&0&1&2&3&1&1&0\\
0&0&0&0&0&-1&-3&-3&-4&-1\\
0&0&0&0&0&0&1&3&6&4\\
0&0&0&0&0&0&0&-1&-4&-6\\
0&0&0&0&0&0&0&0&1&4\\
0&0&0&0&0&0&0&0&0&-1\\
\end{bmatrix},
\end{equation}
which contains $M_4$ as a sub-block in the upper left corner. We can also calculate the inverse of $M_{10}$ to be
\begin{equation}
    {M}_{10}^{-1}=\begin{bmatrix}
1&0&0&0&0&0&0&0&0&0\\
0&-1&-1&-1&-1&-1&-1&-1&-1&-1\\
0&0&1&1&1&1&1&1&1&1\\
0&0&0&-1&-2&-3&-4&-5&-6&-7\\
0&0&0&0&1&2&3&4&5&6\\
0&0&0&0&0&-1&-3&-6&-10&-15\\
0&0&0&0&0&0&1&3&6&10\\
0&0&0&0&0&0&0&-1&-4&-10\\
0&0&0&0&0&0&0&0&1&4\\
0&0&0&0&0&0&0&0&0&-1\\
\end{bmatrix}.
\end{equation}
We should expect in general that the entries of $M_n$ and $M_n^{-1}$ are given in terms of binomial coefficients.
\end{example}
\subsection{Solving the Matrix System}
Here we show that the combinations $A_i$ can be expressed explicitly by the $i$th row of the matrix-vector product $M_n\mathbf{v}$ for any $n$.

\begin{prop}\label{prop:matrix}
Denote $\mathbf{e}\in\mathbb{R}^n$ as the vector of all ones and $\mathbf{v}=(v_1,\ldots,v_n)$, where $v_j=\beta_j/\beta_0$ for each $j$. Consider the expression for $\mathcal{B}_k$ from equation \ref{eq:ansatz}, and replace $|a_k|$ with $1-|b_k|^2$. Then the integral linear combinations $A_i$ from Proposition \ref{prop:necessary} are given by the $i$th entry of the vector $M_n\mathbf{v}$.

\end{prop}
\begin{proof}
Using the binomial theorem on $(-|a_k|^2)^{\left(j-1\right)}$ yields
$$(-|a_k|^2)^{(j-1)}=(|b_k|^2-1)^{j-1}=\sum _{\ell=0}^{j-1}\binom{j-1}{\ell}(-1)^{\ell}|b_k|^{2\left(j-1-\ell\right)}.$$
Plugging this into the sum expression for the ansatz gives us
\begin{equation*}
\begin{aligned}
    &\implies\mathcal{B}_k=\beta_0+|a_k|^2\left(\sum_{j=1}^{\lfloor{n/2}\rfloor}(-|b_k|^2|a_k|^2)^{j-1}(\beta_{2j-1}-|b_k|^2\beta_{2j})\right)\\
    &=\beta_0\left[1+(1-|b_k|^2)\sum_{j=1}^{\lfloor{n/2}\rfloor}\sum _{\ell=0}^{j-1}\binom{j-1}{\ell}(-1)^{\ell}(|b_k|^2)^{2(j-1)-\ell}(v_{2j-1}-|b_k|^2v_{2j})\right].
\end{aligned}
\end{equation*}
From this we can see that in general $[{M}_n]_{ii}=1$ for odd $i$ and $-1$ for even $i$ (this corresponds to $\ell=0$ in the summation above). Consider the terms in this summation with a factor $\beta_{2m-1}/\beta_0$ in front, where $m$ is an integer. Then $m=j$ in this sum, and so the terms with this are of the form $v_{2m-1}\binom{m-1}{\ell}(-1)^\ell (|b_k|^2)^{2\left(m-1\right)-\ell}$, where $\ell=0,\ldots,m-1$. For terms with a factor $v_{2m}$, we have $m=j$ again and the form of these terms are $v_{2m}\binom{j-1}{\ell}(-1)^{\ell+1} (|b_k|^2)^{2(m-1)-\ell}$. Therefore the matrix element $[M_n]_{ij}$ describes the term with the factor $v_j|b_k|^{2(i-1)}$, and so we have $[M_n]_{(2m-1-\ell),(2m-1)}=(-1)^\ell\binom{m-1}{\ell}$ and $[M_n]_{(2m-\ell),(2m)}=(-1)^{\ell+1}\binom{m-1}{\ell}$. Therefore $M_n$ gives the set of coefficients in the combinations $A_i$, with the $i$th row corresponding to terms with the factor $|b_k|^{2\left(i-1\right)}$. Therefore $M_n\mathbf{v}=(A_1,\ldots,A_n)$.
\end{proof}
By Proposition \ref{prop:necessary}, we can see that $M_n\mathbf{v}=(A_1,\ldots,A_n)=(-1,\ldots,-1)=-\mathbf{e}$, and so $\mathcal{B}_k=\beta_0|b_k|^{2n}$ if and only if $M_n\mathbf{v}=-\mathbf{e}$, which is equivalent to $\mathbf{v}=-M_n^{-1}\mathbf{e}$. Now we determine what the inverse matrix $M_n^{-1}$ is.
\begin{lemma}\label{lem:inv}
The inverse matrix ${M}_n^{-1}$ has the following elements: $[{M}_n^{-1}]_{jj}=(-1)^{j+1}$ and the subsequent rows can be described as, for $i\in\mathbb{N}$ and $j\geq 0$:
\begin{gather}
    [{M}_n^{-1}]_{(2i),(2i+j)}=-\binom{j+i-1}{j},\\
    [{M}_n^{-1}]_{(2i+1),(2i+1+j)}=\binom{j+i-1}{j}.
\end{gather}
The rest of the elements are 0.
\end{lemma}
\begin{proof}
We can also rewrite the elements of the inverse by substituting $j\to n-2i-j$:
\begin{align*}
    [{M}_n^{-1}]_{(2i),(n-j)}&=-\binom{n-i-j-1}{n-2i-j}=-\binom{n-i-j-1}{i-1},\\
    [{M}_n^{-1}]_{(2i+1),(n-j)}&=\binom{n-i-j-2}{n-2i-j-1}=\binom{n-i-j-2}{i-1}.
\end{align*}
Here we used the binomial identity $\binom{n}{k}=\binom{n}{n-k}$. We want to show that the product of the matrices is $X={M}_n^{-1}{M}_n=I_n$, where $I_n$ is the $n\times n$ identity matrix. For an inductive argument, assume that the formula is true for ${M}_n$. Then we can write ${M}_{n+1}$ as
$${M}_{n+1}=\begin{bmatrix}
{M}_n&*\\
0&(-1)^{n}
\end{bmatrix},$$
where $*$ is all the extra elements that would be included above the diagonal in the last column. Note that if the formula is valid for $M_n$, then it must be true for $M_{n'}$ for $n'<n$. Therefore the base cases $n=1,2,3,4$ have already been verified by Example \ref{ex:n4matrix}. The inverse is also an upper triangular matrix, and we know that the product of upper triangular matrices are upper triangular, so the inverse looks like 
$${M}_{n+1}^{-1}=\begin{bmatrix}
{M}_n^{-1}&*\\
0&(-1)^{n}
\end{bmatrix}.$$
And of course, the last element on the diagonal is $X_{n+1,n+1}=1$ since the negative signs cancel. It remains to just show the last column above the diagonal is zeros. We split this into two cases, whether $n$ is odd or even. In either case, the first entry in the last column yields $X_{1,n+1}=0$ because the only nonzero entry in the first row of $M_{n+1}^{-1}$ is the first entry, but the first entry of the last column in $M_{n}$ is always 0. Now we split the remaining entries into two cases.
\begin{enumerate}
    \item If $n=2m$ for some integer $m$, then the last column comes in the form
    $$[{M}_{n+1}]_{(2m+1-j),(2m+1)}=(-1)^{j}\binom{m-1}{j},$$
    and we want to multiply this column by every row in ${M}_{n+1}^{-1}$. For even rows $2i$, this looks like
    \begin{equation*}
        \begin{aligned}
            X_{2i,n+1}&=-\sum_{j=0}^{n}(-1)^{j}\binom{m-1}{j}\binom{n-i-j}{i-1}\\
            &=-\sum_{j=0}^{m-1}(-1)^{j}\binom{m-1}{j}\binom{2m-i-j}{i-1}.
        \end{aligned}
    \end{equation*}
    For odd rows $2i+1$ it looks like
    \begin{equation*}
        \begin{aligned}
            X_{2i+1,n+1}&=\sum_{j=0}^{n}(-1)^{j}\binom{m-1}{j}\binom{n-i-j-1}{i-1}\\
            &=\sum_{j=0}^{m-1}(-1)^{j}\binom{m-1}{j}\binom{2m-i-j-1}{i-1}.
        \end{aligned}
    \end{equation*}
    In both cases we are looking at $1\leq i<m$. 
    \item If $n=2m-1$ for some integer $m$, then we get very similar situations:
    $$[{M}_{n+1}]_{(2m-j),(2m)}=(-1)^{j+1}\binom{m-1}{j},$$
    and we want to multiply this column by every row in ${M}_{n+1}^{-1}$. For even rows $2i$, this looks like
    \begin{equation*}
        \begin{aligned}
            X_{2i,n+1}&=-\sum_{j=0}^{n}(-1)^{j+1}\binom{m-1}{j}\binom{n-i-j}{i-1}\\
            &=\sum_{j=0}^{m-1}(-1)^{j}\binom{m-1}{j}\binom{2m-i-j-1}{i-1}.
        \end{aligned}
    \end{equation*}
    For odd rows $2i+1$ it looks like
    \begin{equation*}
        \begin{aligned}
            X_{2i+1,n+1}&=\sum_{j=0}^{n}(-1)^{j+1}\binom{m-1}{j}\binom{n-i-j-1}{i-1}\\
            &=-\sum_{j=0}^{m-1}(-1)^{j}\binom{m-1}{j}\binom{2m-i-j-2}{i-1}.
        \end{aligned}
    \end{equation*}
    Note that the second binomials in each sum are of the form $\binom{\mu+\nu-1}{\mu-1}=\binom{\mu+\nu-1}{\nu}=\mu(\mu+1)\ldots(\mu+\nu-1)/\nu!$, which are polynomials in $\mu$ of degree $\nu$. In each of the sums above, $\nu=i-1$ and $\mu=2m-i-j-R+2$, where $R=0,1,$ or 2 depending on the sum. Therefore $\binom{\mu+\nu-1}{\nu}$ is a polynomial in $j$ of degree $i-1$.
    If $P(j)$ is a polynomial in $j$ with degree less than or equal to $d$, then\cite{Ruiz1996}
    \begin{equation}\label{eq:polynom}
    {\displaystyle \sum _{j=0}^{d}(-1)^{j}{\binom {d}{j}}P(j)=0.}
    \end{equation}
    Substitute $d= m-1$ and $P(j)=\binom{2m-i-j-R}{i-1}$. Since $P(j)$ is a polynomial with degree $i-1\leq m-1$, all of the sums above are equal to 0, and hence the last column is 0 above the diagonal whether $n$ is odd or even, and so by induction $X=I_n$ holds for all $n$. So ${M}_n^{-1}$ is the inverse of the matrix ${M}_n$.
\end{enumerate}
\end{proof}
From this we can solve the system of equations $M_n\mathbf{v}=-\mathbf{e}$ for $\mathbf{v}$, given as $\mathbf{v}=-M_n^{-1}\mathbf{e}$.
\begin{thm}\label{cor:values_vi}
For an order $n$, the given diagonalizing sequence has the property of Conjecture \ref{conj:convergence} if and only if for all $j$ the coefficients $\beta_j$ from Proposition \ref{prop:ansatz} take the values $\beta_1=-\beta_0$, $\beta_{2i}=\beta_0\binom{n-i}{i}$, and $\beta_{2i+1}=-\beta_0\binom{n-i-1}{i}$, where $i\in\mathbb{N}$.
\end{thm}
\begin{proof}
By Proposition \ref{prop:necessary}, the conjectured property $|b_{k+1}|=|b_k|^{2n+1}$ holds if and only if each combination $A_i=-1$, which is equivalent to $M_n\mathbf{v}=-\mathbf{e}$ since $A_i$ is the $i$th entry of $M_n\mathbf{v}$ by Proposition \ref{prop:matrix}.
The vector $\mathbf{v}$ is given by $-M_n^{-1}\mathbf{e}$, with $M_n^{-1}$ given by Lemma \ref{lem:inv}. We can use the well-known binomial identities $\sum_{n=0}^{N}\binom{m+n}{n}=\binom{N+m+1}{N}$ and $\binom{n}{k}=\binom{n}{n-k}$ to simplify the expressions:
\begin{equation}\label{eq:binom}
\begin{aligned}
v_1&=-1,\quad v_{2i}=\sum_{j=0}^{n-2i}\binom{j+i-1}{j}=\binom{n-i}{n-2i}=\binom{n-i}{i},\\
v_{2i+1}&=-\sum_{j=0}^{n-2i-1}\binom{j+i-1}{j}=-\binom{n-i-1}{n-2i-1}=-\binom{n-i-1}{i}.
\end{aligned}
\end{equation}
Using the fact that $v_j=\beta_j/\beta_0$ yields the expressions for $\beta_j$ in the theorem statement.
\end{proof}
This theorem establishes the direct equivalence between the property from Conjecture \ref{conj:convergence} and the exact numerical values that the $\beta_j$'s take.
\section{Using Recursive Equations to Solve Basic Sequences}\label{section5}
We can directly compute formulas for $\beta_j$ and $\alpha_j$ in terms of $\lambda_{i0}$ and $\lambda_{i1}$ for $i=1,\ldots,n$ and a given $n$ by expanding $U_{k+1}^{(n)}$ through matrix multiplication. However, in general it is much simpler to use the recursive equations \ref{eq:alpha} and \ref{eq:beta}. Here we compute some basic cases ($j=0,1$, and $n$), and at the end we will show the explicit formulas for $\beta_j$ and $\alpha_j$ for arbitrary $j$ and $n$ in terms of the $\lambda_{jl}$'s in Theorem \ref{thm:general_alpha_beta}. We also show in these basic cases that they take on the values given from Theorem \ref{cor:values_vi}, confirming the conjecture at least partially.


For $j=0$, the recursive equations \eqref{eq:alpha} and \eqref{eq:beta} reduce to
\begin{align*}
    \alpha_0^{(n+1)}&=\lambda_{10}^2{T}{\alpha}_{0}^{(n)},\\
    \beta_0^{(n+1)}&=\lambda_{10}\lambda_{11}{T}{\beta}_0^{(n)},
\end{align*}
and here we will solve these equations.
\begin{thm}\label{eq:j0}
For any order $n$, $\alpha_0$ and $\beta_0$ are explicitly given by the formulas
\begin{equation}
    \alpha_0=\prod_{k=1}^{n}\lambda_{k0}^2=\omega^{\frac{1}{2}(2 (-1)^n n + (-1)^n - 1)},\quad{\beta}_0 =\prod_{k=1}^{n}\lambda_{k0}\lambda_{k1}=(-1)^{\frac{n\left(n+3\right)}{2}}.
\end{equation}
\end{thm}

\begin{proof}
For $n=1$, explicitly calculating the sequence gives $\beta_0=\lambda_{10}\lambda_{11}$ and $\alpha_0=\lambda_{10}^2$, and so the above satisfies the initial conditions. Applying the shift $T$, we have
$$\lambda_{10}^2{T}{\alpha}_0^{(n)}=\lambda_{10}^2\prod_{k=2}^{n+1}\lambda_{k0}^2=\prod_{k=1}^{n+1}\lambda_{k0}^2={\alpha}_0^{(n+1)}.$$
$$\lambda_{10}\lambda_{11}{T}{\beta}_0^{(n)}= \lambda_{10}\lambda_{11}\prod_{k=2}^{n}\lambda_{k0}\lambda_{k1}=\prod_{k=1}^{n+1}\lambda_{k0}\lambda_{k1}={\beta}_0^{(n+1)}.$$
Therefore the formulas above solve the recursive equations. Since $\lambda_{j0}=\omega^{\left(-1\right)^{j}j}$ and $\lambda_{j1}=(-1)^{j+1}\omega^{\left(-1\right)^{j+1}j}$, we have that
\begin{gather*}
    \alpha_0=\prod_{k=1}^{n}\omega^{2(-1)^kk}=\omega^{2\sum_{k=1}^{n}(-1)^kk}=\omega^{\frac{1}{2}(2 (-1)^n n + (-1)^n - 1)},\\
    {\beta}_0 =\prod_{k=1}^{n}(-1)^{k+1}=(-1)^{\sum_{k=1}^{n}k+1}=(-1)^{\frac{n\left(n+3\right)}{2}}.
\end{gather*}
\end{proof}
For $n=1,2,3,4$, we have $\alpha_0^{(n)}=\omega^{-2},\omega^{2},\omega^{-4},\omega^{4}$ and $\beta_0^{(n)}=1,-1,-1,1$, and the pattern continues in this fashion. For $m\in\mathbb{N}$, we can see that $\alpha_0^{(2m)}=\omega^{2m}$, $\alpha_0^{(2m-1)}=\omega^{-2m}$, $\beta_0^{(2m)}=(-1)^m$, and $\beta_0^{(2m-1)}=(-1)^{m+1}$. 

The next is $j=n$, which is simple to examine as well. For any pair $(n,j)$, we note that $\alpha_j^{(n)}=\beta_j^{(n)}=0$ for $j>n$.
\begin{align*}
    \alpha_{n+1}^{(n+1)}&=\lambda_{10}^2{T}{\alpha}_{n+1}^{(n)}+\lambda_{10}^2({T}{\alpha}_{n}^{(n)}-\chi_{n+1}{T}{\alpha}_{n+1}^{(n)})\\
    &+\lambda_{11}^2\left(\overline{{T}{\alpha}_{n}^{(n)}}-\chi_{n+1}\overline{{T}{\alpha}_{n+1}^{(n)}}\right)-2\lambda_{10}\lambda_{11}{T}{\beta}_{n}^{(n)}\\
    &=\lambda_{10}^2{T}{\alpha}_{n}^{(n)}+\lambda_{11}^2\overline{{T}{\alpha}_{n}^{(n)}}-2\lambda_{10}\lambda_{11}{T}{\beta}_{n}^{(n)},\\
    \beta_{n+1}^{(n+1)}&=\lambda_{10}\lambda_{11}{T}{\beta}_{n+1}^{(n)}+\lambda_{10}^2{T}{\alpha}_{n}^{(n)}+\lambda_{11}^2\overline{{T}{\alpha}_{n}^{(n)}}-2\lambda_{10}\lambda_{11}\left({T}{\beta}_{n}^{(n)}+\chi_{n+1}{T}{\beta}_{n+1}^{(n)}\right)\\
    &=\lambda_{10}^2{T}{\alpha}_{n}^{(n)}+\lambda_{11}^2\overline{{T}{\alpha}_{n}^{(n)}}-2\lambda_{10}\lambda_{11}{T}{\beta}_{n}^{(n)}.
\end{align*}
All of the cancelling above comes from the assumption that we gave before, that all functions are zero for $j>n$.
We claim that both $\alpha_n$ and $\beta_n$ take the same form for order $n$.
\begin{thm}\label{eq:jn}
For any order $n$, both $\beta_n$ and $\alpha_n$ are given by the same expression
\begin{equation}
    \alpha_n=\beta_n=\prod_{k=1}^{n}(\lambda_{k0}-\lambda_{k1})^2=(-1)^{\frac{n\left(n+1\right)}{2}}=(-1)^n\beta_0.
\end{equation}
\end{thm}
\begin{proof}
Note that $(\lambda_{j0}-\lambda_{j1})^2=\lambda_{j0}^2+\lambda_{j1}^2-2\lambda_{j0}\lambda_{j1}$ is real, as $\lambda_{j0}^2=\omega^{2(-1)^jj}$ is the conjugate of $\lambda_{j1}^2=\omega^{2(-1)^{j+1}j}$, and $\lambda_{j0}\lambda_{j1}=(-1)^{j+1}$. The base case $\alpha_1=\beta_1=(\lambda_{10}-\lambda_{11})^2$ confirms the theorem for $n=1$, and by induction we can see that $\alpha_n=\beta_n$ implies that $\alpha_{n+1}=\beta_{n+1}$, therefore $\alpha_n=\beta_n$ for all $n$. The recursive equations give us
\begin{equation*}
    \begin{aligned}
        \beta_{n+1}&=\lambda_{10}^2{T}{\alpha}_{n}+\lambda_{11}^2\overline{{T}{\alpha}_{n}}-2\lambda_{10}\lambda_{11}{T}{\beta}_{n}\\
        &=(\lambda_{10}^2+\lambda_{11}^2-2\lambda_{10}\lambda_{11})\prod_{k=2}^{n+1}(\lambda_{k0}-\lambda_{k1})^2\\
        &=\prod_{k=1}^{n+1}(\lambda_{k0}-\lambda_{k1})^2.
    \end{aligned}
\end{equation*}
Therefore the expression above satisfies the recursive equations. Plugging in the roots of unity gives the identity
\begin{equation*}
    \begin{aligned}
        \prod_{k=1}^{n}(\lambda_{k0}-\lambda_{k1})^2&=\prod_{k=1}^{n}\left(\omega^{(-1)^{k}k}+(-1)^k\omega^{(-1)^{k+1}k}\right)^2\\
        &=\prod_{k=1}^{n}\left(\omega^{2(-1)^{k}k}+\omega^{2(-1)^{k+1}k}+2(-1)^k\right)\\
        &=\prod_{k=1}^{n}\left(2\cos(k\theta)+2(-1)^k\right)=\left(-1\right)^{\frac{n\left(n+1\right)}{2}}.
    \end{aligned}
\end{equation*}
The identity above is proved in \ref{app:trig}. Dividing by $\beta_0$ is the same as multiplying by it, and so we get that $v_n$ is equal to
$$v_n=\beta_n/\beta_0=(-1)^{\frac{n\left(n+1\right)}{2}}(-1)^{\frac{n\left(n+3\right)}{2}}=(-1)^n.$$
\end{proof}
Now we consider $j=1$, in which case our inductive sequences are
\begin{align*}
    \alpha_1^{(n+1)}&=\lambda_{10}^2{T}{\alpha}_{0}^{(n)}+\lambda_{11}^2\overline{{T}{\alpha}_{0}^{(n)}}-2\lambda_{10}\lambda_{11}{T}{\beta}_{0}^{(n)}-\lambda_{11}^2\overline{{T}{\alpha}_{1}^{(n)}},\\
    \beta_1^{(n+1)}&=\lambda_{10}^2{T}{\alpha}_{0}^{(n)}+\lambda_{11}^2\overline{{T}{\alpha}_{0}^{(n)}}-2\lambda_{10}\lambda_{11}{T}{\beta}_{0}^{(n)}-\lambda_{10}\lambda_{11}{T}{\beta}_1^{(n)}.
\end{align*}
However, we already know the first three terms using our results from Theorem \ref{eq:j0} for $\alpha_0$ and $\beta_1$, so we can write
\begin{align*}
    \alpha_1^{(n+1)}&=\prod_{k=1}^{n+1}\lambda_{k0}^2+\prod_{k=1}^{n+1}\lambda_{k1}^2-2\prod_{k=1}^{n+1}\lambda_{k0}\lambda_{k1}-\lambda_{11}^2\overline{{T}{\alpha}_{1}^{(n)}}\\
    &=\left(\prod_{k=1}^{n+1}\lambda_{k0}-\prod_{k=1}^{n+1}\lambda_{k1}\right)^2-\lambda_{11}^2\overline{{T}{\alpha}_{1}^{(n)}},\\
    \beta_1^{(n+1)}&=\prod_{k=1}^{n+1}\lambda_{k0}^2+\prod_{k=1}^{n+1}\lambda_{k1}^2-2\prod_{k=1}^{n+1}\lambda_{k0}\lambda_{k1}-\lambda_{10}\lambda_{11}{T}{\beta}_1^{(n)}\\
    &=\left(\prod_{k=1}^{n+1}\lambda_{k0}-\prod_{k=1}^{n+1}\lambda_{k1}\right)^2-\lambda_{10}\lambda_{11}{T}{\beta}_1^{(n)}.
\end{align*}
Here we will solve these equations for $\alpha_1$ and $\beta_1$.
\begin{thm}
For any order $n$, $\alpha_1$ and $\beta_1$ are given by the formulas
\begin{equation}
    \alpha_1
    =\left(\prod_{k=1}^{n}\lambda_{k0}-\prod_{k=1}^{n}\lambda_{k1}\right)^2+\sum_{\ell=1}^{n-1}(-1)^{\ell}\prod_{k=1}^{\ell}\lambda_{k\chi_{k}}^{2}\left(\prod_{k=\ell+1}^{n}\lambda_{k0}-\prod_{k=\ell+1}^{n}\lambda_{k1}\right)^2,
\end{equation}
\begin{equation}
    \beta_1
    =\left(\prod_{k=1}^{n}\lambda_{k0}-\prod_{k=1}^{n}\lambda_{k1}\right)^2+\sum_{\ell=1}^{n-1}(-1)^{\ell}\prod_{k=1}^{\ell}\lambda_{k0}\lambda_{k1}\left(\prod_{k=\ell+1}^{n}\lambda_{k0}-\prod_{k=\ell+1}^{n}\lambda_{k1}\right)^2.
\end{equation}
\end{thm}
\begin{proof}
Since the top index of the sum is $0$ for $n=1$, we can ignore the sum, which would correspond to $\alpha_1=\beta_1=(\lambda_{10}-\lambda_{11})^2$, as we computed for Proposition \ref{prop:ansatz}. Applying the shift $T$ to each term in the recursive equation gives us the following expansion for the $\lambda_{11}^2\overline{T\alpha_1^{(n)}}$:
\begin{align*}
    \lambda_{11}^2\overline{{T}{\alpha}_1^{(n)}}
    &=\lambda_{11}^2\left(\prod_{k=2}^{n+1}\lambda_{k0}-\prod_{k=2}^{n+1}\lambda_{k1}\right)^2\\&+\lambda_{11}^2\sum_{\ell=1}^{n-1}(-1)^{\ell}\prod_{k=2}^{\ell+1}\lambda_{k\chi_{k}}^{2}\left(\prod_{k=\ell+2}^{n+1}\lambda_{k0}-\prod_{k=\ell+2}^{n+1}\lambda_{k1}\right)^2\\
    &=\sum_{\ell=1}^{n}(-1)^{\ell-1}\prod_{k=1}^{\ell}\lambda_{k\chi_{k}}^{2}\left(\prod_{k=\ell+1}^{n+1}\lambda_{k0}-\prod_{k=\ell+1}^{n+1}\lambda_{k1}\right)^2.
\end{align*}
Applying the same approach to $\lambda_{10}\lambda_{11}{T}{\beta}_1^{(n)}$ gives the same sort of expansion:
\begin{align*}
    \lambda_{10}\lambda_{11}{T}{\beta}_1^{(n)}
    &=\lambda_{10}\lambda_{11}\left(\prod_{k=2}^{n+1}\lambda_{k0}-\prod_{k=2}^{n+1}\lambda_{k1}\right)^2\\&+\lambda_{10}\lambda_{11}\sum_{\ell=1}^{n-1}(-1)^{\ell}\prod_{k=2}^{\ell+1}\lambda_{k0}\lambda_{k1}\left(\prod_{k=\ell+2}^{n+1}\lambda_{k0}-\prod_{k=\ell+2}^{n+1}\lambda_{k1}\right)^2\\
    &=\sum_{\ell=1}^{n}(-1)^{\ell-1}\prod_{k=1}^{\ell}\lambda_{k0}\lambda_{k1}\left(\prod_{k=\ell+1}^{n+1}\lambda_{k0}-\prod_{k=\ell+1}^{n+1}\lambda_{k1}\right)^2.
\end{align*}
Now we can take these results and plug them into the equations for $\alpha_1^{(n+1)}$ and $\beta_1^{(n+1)}$:
\begin{align*}
    \alpha_1^{(n+1)}
    &=\left(\prod_{k=1}^{n+1}\lambda_{k0}-\prod_{k=1}^{n+1}\lambda_{k1}\right)^2-\lambda_{11}^2\overline{{T}{\alpha}_{1}^{(n)}}\\
    &=\left(\prod_{k=1}^{n+1}\lambda_{k0}-\prod_{k=1}^{n+1}\lambda_{k1}\right)^2+\sum_{\ell=1}^{n}(-1)^{\ell}\prod_{k=1}^{\ell}\lambda_{k\chi_{k+1}}^{2}\left(\prod_{k=\ell+1}^{n+1}\lambda_{k0}-\prod_{k=\ell+1}^{n+1}\lambda_{k1}\right)^2,
\end{align*}
\begin{align*}
    \beta_1^{(n+1)}
    &=\left(\prod_{k=1}^{n+1}\lambda_{k0}-\prod_{k=1}^{n+1}\lambda_{k1}\right)^2-\lambda_{10}\lambda_{11}{T}{\beta}_1^{(n)}\\
    &=\left(\prod_{k=1}^{n+1}\lambda_{k0}-\prod_{k=1}^{n+1}\lambda_{k1}\right)^2+\sum_{\ell=1}^{n}(-1)^{\ell}\prod_{k=1}^{\ell}\lambda_{k0}\lambda_{k1}\left(\prod_{k=\ell+1}^{n+1}\lambda_{k0}-\prod_{k=\ell+1}^{n+1}\lambda_{k1}\right)^2.
\end{align*}
And this confirms our formula by induction. 
\end{proof}
\begin{remark}
In abuse of notation, introduce variables $\lambda_{00}=\lambda_{01}=1$, and so we can write the formulas as 
\begin{equation}\label{alphaj1}
    \alpha_1^{(n)}
    =\sum_{\ell=0}^{n-1}(-1)^{\ell}\prod_{k=0}^{\ell}\lambda_{k\chi_{k}}^{2}\left(\prod_{k=\ell+1}^{n}\lambda_{k0}-\prod_{k=\ell+1}^{n}\lambda_{k1}\right)^2,
\end{equation}
\begin{equation}\label{betaj1}
    \beta_1^{(n)}
    =\sum_{\ell=0}^{n-1}(-1)^{\ell}\prod_{k=0}^{\ell}\lambda_{k0}\lambda_{k1}\left(\prod_{k=\ell+1}^{n}\lambda_{k0}-\prod_{k=\ell+1}^{n}\lambda_{k1}\right)^2.
\end{equation}
Applying the shifting operation doesn't really make sense here since $\lambda_{00}$ and $\lambda_{01}$ are not variables the shifting operator acts on. However, if we multiply by $\lambda_{10}\lambda_{11}$ or what respective factor is in front of the shift, all is well since this takes the place of the $\lambda_{00},\lambda_{01}$, which are now gone:
\begin{align*}
    \lambda_{11}^2\overline{{T}{\alpha}_1^{(n)}}
    &=\sum_{\ell=1}^{n}(-1)^{\ell-1}\prod_{k=1}^{\ell}\lambda_{k\chi_{k}}^{2}\left(\prod_{k=\ell+1}^{n+1}\lambda_{k0}-\prod_{k=\ell+1}^{n+1}\lambda_{k1}\right)^2,\\
    \lambda_{10}\lambda_{11}{T}{\beta}_1^{(n)}
    &=\sum_{\ell=1}^{n}(-1)^{\ell-1}\prod_{k=1}^{\ell}\lambda_{k0}\lambda_{k1}\left(\prod_{k=\ell+1}^{n+1}\lambda_{k0}-\prod_{k=\ell+1}^{n+1}\lambda_{k1}\right)^2.
\end{align*}
We will use this notation and shifting rule when deriving the general expression for $\beta_j$ and $\alpha_j$.
\end{remark}

Plugging in for $\lambda_{j0}$ and $\lambda_{j1}$ gives the expression
\begin{gather*}
    \beta_1
    =\left(\prod_{k=1}^{n}\omega^{2(-1)^{k}k}+\prod_{k=1}^{n}\omega^{2(-1)^{k+1}k}-2(-1)^{\frac{n\left(n+3\right)}{2}}\right)\\
    +\sum_{\ell=1}^{n-1}(-1)^{\frac{\ell\left(\ell+1\right)}{2}}\left(\prod_{k=\ell+1}^{n}\omega^{2(-1)^{k}k}+\prod_{k=\ell+1}^{n}\omega^{2(-1)^{k+1}k}-2(-1)^{\frac{n\left(n+3\right)}{2}-\frac{\ell\left(\ell+3\right)}{2}}\right)\\
    =\sum_{\ell=0}^{n-1}(-1)^{\frac{\ell\left(\ell+1\right)}{2}}\left(\prod_{k=\ell+1}^{n}\omega^{2(-1)^{k}k}+\prod_{k=\ell+1}^{n}\omega^{2(-1)^{k+1}k}-2(-1)^{\frac{n\left(n+3\right)}{2}-\frac{\ell\left(\ell+3\right)}{2}}\right).
\end{gather*}
Now we confirm that the expression for $\beta_1$ simplifies to $\beta_1/\beta_0=-1$ from Theorem \ref{cor:values_vi}.
\begin{thm}\label{eq:j1}
For any order $n$, we have $\beta_1/\beta_0=-1$.
\end{thm}
\begin{proof}
Note that $(-1)^{\frac{\ell\left(\ell+1\right)}{2}-\frac{\ell\left(\ell+3\right)}{2}}=(-1)^\ell$. We can consider this identity in two cases. For $n=2m$, note that we can simplify the last term as follows:
\begin{equation*}
    (-1)^{\frac{n\left(n+3\right)}{2}}=(-1)^{\frac{2m\left(2m+3\right)}{2}} = (-1)^{m}.
\end{equation*}
Therefore we can simplify the above expression for $\beta_1$ to
\begin{gather*}
    \implies\beta_1^{(2m)}
    =\sum_{\ell=0}^{2m-1}(-1)^{\frac{\ell\left(\ell+1\right)}{2}}\left(\prod_{k=\ell+1}^{2m}\omega^{2(-1)^{k}k}+\prod_{k=\ell+1}^{2m}\omega^{2(-1)^{k+1}k}\right)-2(-1)^{m+\ell}.
\end{gather*}
Since we are summing over an even number of terms, we have that $\sum_{\ell=0}^{2m-1}2(-1)^{m+\ell}=2(-1)^m\sum_{\ell=0}^{2m-1}(-1)^{\ell}=0$.
Say that $\ell=2l$, then the first two terms simplify to
\[\prod_{k=2l+1}^{2m}\omega^{2(-1)^{k}k}+\prod_{k=2l+1}^{2m}\omega^{2(-1)^{k+1}k}=\omega^{2(m-l)}+\omega^{-2(m-l)}=2\cos((m-l)\theta).\]
The sign in front of these terms is given by $(-1)^{\frac{\ell\left(\ell+1\right)}{2}}=(-1)^{l\left(2l+1\right)}=(-1)^l$.
Similarly for $\ell=2l-1$ we have
\[\prod_{k=2l}^{2m}\omega^{2(-1)^{k}k}+\prod_{k=2l}^{2m}\omega^{2(-1)^{k+1}k}=\omega^{2(m+l)}+\omega^{-2(m+l)}=2\cos((m+l)\theta).\]
The sign in front of these terms is also given by $(-1)^{\frac{\ell\left(\ell+1\right)}{2}}=(-1)^{l\left(2l-1\right)}=(-1)^l$.
Note that we get the term $\cos\theta$ for $\ell = 2m-2$ and $\cos (2m\theta)$ for $\ell = 2m-1$. We have a positive sign in front of $\cos(m'\theta)$ for $m'=m$, and as we increase or decrease $m'$ we alternate signs. Therefore we can write $\beta_1$ as
\[\beta_1^{(2m)}=2(-1)^m\sum_{k=1}^{2m}(-1)^k\cos(k\theta).\]
Since $\beta_0^{(2m)}=(-1)^m$ we have that by Lemma \hyperref[LemmaA1]{A.1}
\[v_1^{(2m)}=\beta_1^{(2m)}/\beta_0^{(2m)}=2\sum_{k=1}^{2m}(-1)^k\cos(k\theta)=-1.\]

For $n=2m-1$, we simplify powers and evaluate the products in the sum. This time we have $(-1)^{\frac{n\left(n+3\right)}{2}}=(-1)^{m+1}$:
\[
    \beta_1^{(2m-1)}
    =\sum_{\ell=0}^{2m-2}(-1)^{\frac{\ell\left(\ell+1\right)}{2}}\left(\prod_{k=\ell+1}^{2m-1}\omega^{2(-1)^{k}k}+\prod_{k=\ell+1}^{2m-1}\omega^{2(-1)^{k+1}k}\right)-2(-1)^{m+\ell+1}.
\]
Similar to before, the alternating sign cancels with itself except for $\ell=0$, which gives $-2(-1)^{m+1}$. Now we consider whether $\ell$ is even or odd, like before. The sign given by $(-1)^{\frac{\ell\left(\ell+1\right)}{2}}$ is the same as before.
For $\ell=2l$, we have
\[\prod_{k=2l+1}^{2m-1}\omega^{2(-1)^{k}k}+\prod_{k=2l+1}^{2m-1}\omega^{2\left(-1\right)^{k+1}k}=\omega^{-2\left(m+l\right)}+\omega^{2\left(m+l\right)}=2\cos(\left(m+l\right)\theta).\]
Similarly for $\ell=2l-1$ we have
\[\prod_{k=2l}^{2m-1}\omega^{2(-1)^{k}k}+\prod_{k=2l}^{2m-1}\omega^{2(-1)^{k+1}k}=\omega^{-2\left(m-l\right)}+\omega^{2\left(m-l\right)}=2\cos(\left(m-l\right)\theta).\]
Once again the remaining terms combine together in the same way. We get a $\cos\theta$ term from $\ell=2m-3$ and a $\cos((2m-1)\theta)$ term from $\ell=2m-2$. Additionally, the $\cos(m'\theta)$ term has a positive sign for $m'=m$, and as we increase or decrease $m'$ the sign alternates. Therefore we rewrite $\beta_1$ as
\begin{align*}
    \beta_1^{(2m-1)}&= -2(-1)^{m+1} + 2(-1)^m\sum_{k=1}^{2m-1}(-1)^k\cos(k\theta)\\
    &=(-1)^{m+1}\left(-2 - 2\sum_{k=1}^{2m-1}(-1)^k\cos(k\theta)\right)\\
    &=2(-1)^{m+1}\sum_{k=0}^{2m-1}(-1)^{k+1}\cos(k\theta).
\end{align*}
Factoring out $\beta_0^{(2m-1)}=(-1)^{m+1}$ yields
\[\beta_1/\beta_0^{(2m-1)}=2\sum_{k=0}^{2m-1}(-1)^{k+1}\cos(k\theta)=-1.\]
We have covered both cases, so this completes the proof.
\end{proof}


\subsection{Applications to n=3}
Here we consider some special cases of the formulas above, namely for order $n=3$, which corresponds to the angle $\theta=\pi/7$. 
Our recursive equations for $j=2$ are the set
\begin{equation}
\begin{aligned}
    \alpha_2^{(n+1)}&=\lambda_{10}^2{T}{\alpha}_{1}^{(n)}+\lambda_{11}^2\overline{{T}{\alpha}_{1}^{(n)}}-2\lambda_{10}\lambda_{11}{T}{\beta}_{1}^{(n)}+\lambda_{10}^2{{T}{\alpha}_{2}^{(n)}},\\
    \beta_2^{(n+1)}&=\lambda_{10}^2{T}{\alpha}_{1}^{(n)}+\lambda_{11}^2\overline{{T}{\alpha}_{1}^{(n)}}-2\lambda_{10}\lambda_{11}{T}{\beta}_{1}^{(n)}+\lambda_{10}\lambda_{11}{T}{\beta}_2^{(n)}.
\end{aligned}    
\end{equation}
By Theorem \ref{eq:j0}, we have $\beta_0=(-1)^{\frac{3(3+3)}{2}}=(-1)^9=-1$.
By Theorems \ref{eq:jn} and \ref{eq:j1} we respectively have that $\beta_1/\beta_0=-1$ and $\beta_3/\beta_0=(-1)^3=-1$, therefore $\beta_1=\beta_3=1$.
In Appendix B we explicitly calculated the formulas for $\beta_0$, $\beta_1$, $\beta_2$, and $\beta_3$ in terms of the $\lambda_{jl}$'s. For $\beta_2$ we have
\begin{equation*}
\begin{aligned}
    \beta_2&=(\lambda_{11}-\lambda_{10})^2(\lambda_{20}\lambda_{30}-\lambda_{21}\lambda_{31})^2-(\lambda_{11}\lambda_{20}-\lambda_{10}\lambda_{21})^2(\lambda_{30}-\lambda_{31})^2\\&+\lambda_{10}\lambda_{11}(\lambda_{20}-\lambda_{21})^2(\lambda_{30}-\lambda_{31})^2.
\end{aligned}
\end{equation*}
Note that $\lambda_{k0}\lambda_{k1}=(-1)^{k+1}$ and $(\lambda_{k0}-\lambda_{k1})^2 = \lambda_{k0}^2+\lambda_{k1}^2-2\lambda_{k0}\lambda_{k1}=\omega^{2k}+\omega^{-2k}-2(-1)^{k+1}=2\cos(k\theta)+2(-1)^k$. The remaining factors can be evaluated as
\begin{align*}
    (\lambda_{20}\lambda_{30}-\lambda_{21}\lambda_{31})^2 = (\omega^2\cdot\omega^{-3}-(-\omega^{-2})\cdot\omega^{3})^2=2\cos\theta+2\\
    (\lambda_{11}\lambda_{20}-\lambda_{10}\lambda_{21})^2=(\omega\cdot\omega^{2}-\omega^{-1}\cdot(-\omega^{-2}))^2=2\cos3\theta+2.
\end{align*}
Therefore $\beta_2$ can be expressed in terms of $\theta$ as
\begin{align*}
    \beta_2&=(2\cos\theta-2)(\lambda_{20}\lambda_{30}-\lambda_{21}\lambda_{31})^2-(\lambda_{11}\lambda_{20}-\lambda_{10}\lambda_{21})^2(2\cos3\theta-2)\\&+(\cos2\theta+2)(\cos3\theta-2),\\
    &=(2\cos\theta-2)(2\cos\theta+2)-(2\cos3\theta-2\cos2\theta)(2\cos3\theta-2).
\end{align*}
Now we evaluate $\beta_2$.
\begin{lemma}
For $\theta=\pi/7$,
\[\beta_2=(2\cos\theta-2)(2\cos\theta+2)-(2\cos3\theta-2\cos2\theta)(2\cos3\theta-2)=-2.\]
\end{lemma}
\begin{proof}
To start, $2\cos3\theta-2\cos2\theta=1-2\cos\theta$ by Lemma \hyperref[LemmaA1]{A.1}, and $(2\cos\theta-2)(2\cos\theta+2)=2\cos2\theta - 2$ by double angle identities.
\begin{align*}
    \beta_2&=(2\cos2\theta - 2) - (1-2\cos\theta)(2\cos3\theta-2)\\
    &=-4\cos\theta + 2\cos2\theta - 2\cos3\theta + 4\cos\theta\cos3\theta\\
    &=-4\cos\theta + 2\cos2\theta - 2\cos3\theta + 2\cos2\theta + 2\cos4\theta\\
    &=-4\cos\theta + 4\cos2\theta - 4\cos3\theta = -2.
\end{align*}
\end{proof}
From these results we get that $(\beta_0,\beta_1,\beta_2,\beta_3)=(-1,1,-2,1)$. This matches Theorem \ref{cor:values_vi} as $\beta_1=-\beta_0$, $\beta_2=\beta_0\binom{3-1}{3-2}=-2$, and $\beta_3=-\beta_0\binom{3-1-1}{3-2-1}=1$.
To confirm the property of Conjecture \ref{conj:convergence} directly, substituting each value for $\beta_j$ gives
\begin{align*}
    \mathcal{B}_k&=\beta_0+|a_k|^2(\beta_1-|b_k|^2(\beta_2+|a_k|^2\beta_3))\\
    &=\beta_0+|a_k|^2(\beta_1+\beta_2(-|b_k|^2)+\beta_3(-|b_k|^2+|b_k|^4)).
\end{align*}
Observe that from Section 4 we showed that the required matrix system to guarantee $\mathcal{B}_k=\beta_0|b_k|^{6}$ is given by $M_3\mathbf{v}=-\mathbf{e}$:
\[
M_3\mathbf{v}=
\begin{bmatrix}
1&0&0\\
0&-1&-1\\
0&0&1
\end{bmatrix}
\begin{bmatrix}
\beta_1/\beta_0\\\beta_2/\beta_0\\\beta_3/\beta_0
\end{bmatrix}
=
\begin{bmatrix}
-1\\-1\\-1
\end{bmatrix}\implies
\begin{bmatrix}
\beta_1/\beta_0\\\beta_2/\beta_0\\\beta_3/\beta_0
\end{bmatrix}=
\begin{bmatrix}
-1\\2\\-1
\end{bmatrix}.\]
The coefficients $\beta_j$ equal the solution to this system, therefore we should expect that $\mathcal{B}_k$ has the desired property from Conjecture \ref{conj:convergence}. And indeed it does:
\begin{align*}
    \mathcal{B}_k&=-1+|a_k|^2(1-2(-|b_k|^2)+(-|b_k|^2+|b_k|^4))\\
    &=-1+(1-|b_k|^2)(1+|b_k|^2+|b_k|^4)=-1+(1-|b_k|^6)=-|b_k|^6.
\end{align*}
Therefore we get $|(U_{k+1})_{21}|=|b_k\mathcal{B}_k|=|b_k|\cdot |b_k|^{6}=|b_k|^7$, which is our desired relation. 
By multiplying by the constant $(\alpha_0)^{-1}$ and writing $D(\theta)=\diag(1,e^{i\theta})$ we can also construct the sequence
\begin{equation}
    U_{k+1}=U_k D(\theta) U_k^{-1}D(\theta)^{5}U_k D(\theta)^{3}U_k^{-1} D(\theta)^{3}U_k D(\theta)^{5}U_k^{-1} D(\theta)U_k.
\end{equation}
By Proposition \ref{prop:phase}, this sequence has the same property because we distribute a $\lambda_{j0}^{-1}$ to each $D_j$, giving $D_j\lambda_{j0}^{-1}=\diag(1,(-1)^{j+1}e^{i\theta_p(-1)^{j+1}j})$. So $D_1\lambda_{10}^{-1}=D(\theta)$, $D_2\lambda_{20}^{-1}=D(\theta)^{5}$, and $D_3\lambda_{30}^{-1}=D(\theta)^{3}$ by using the identity $e^{j\pi i}=(-1)^j$. Since these sequences differ by an additional phase multiplied at the end, they both have the conjectured property by Proposition \ref{prop:phase}.

\subsection{Arbitrary Order}
Now we derive the complete solution to the recursive equations \ref{eq:alpha} and \ref{eq:beta}
We give the following result, which describes the complete solution to these equations.
\begin{thm}\label{thm:general_alpha_beta}
The general solution for $\alpha_j^{(n)}$ and $\beta_j^{(n)}$ for $j\leq n$ is given by the following. Denote $L_{j_1}^{j_2}=\sum_{j=j_1}^{j_2}\ell_j$, where $\ell_j$ are indices of summation. Then for $j\geq 1$,
\begin{equation}\label{alphaj_arb}
\begin{aligned}
    \alpha_j^{(n)}&=\sum_{\ell_j=0}^{n-j}\left[\prod_{k=0}^{\ell_j}\lambda_{k,\chi_j\chi_{k}}^2\right]\sum_{\ell_{j-1}=0}^{n-j-L^{j}_{j}}\sum_{\ell_{j-2}=0}^{n-j-L^{j}_{j-1}}\ldots\sum_{\ell_{1}=0}^{n-j-L^{j}_{2}}\left(\prod_{k=L_{1}^{j}+j}^{n}\lambda_{k0}-\prod_{k=L_{1}^{j}+j}^{n}\lambda_{k1}\right)^2\\
    &\times\left[\prod_{\mu=1}^{j}(-1)^{\chi_{\mu}\ell_{\mu}}\right]\times\left[\prod_{\mu=1}^{j-1}\left(\prod_{k=L^{j}_{1+\mu}+j-\mu}^{L^{j}_{\mu}+j-\mu}\lambda_{k,(\chi_\mu\chi_k)} - \prod_{k=L^{j}_{1+\mu}+j-\mu}^{L^{j}_{\mu}+j-\mu}\lambda_{k,(1-\chi_\mu\chi_k)}\right)^2\right],
\end{aligned}
\end{equation}
\begin{equation}\label{betaj_arb}
\begin{aligned}
    \beta_j^{(n)}&=\sum_{\ell_j=0}^{n-j}\left[\prod_{k=0}^{\ell_j}\lambda_{k0}\lambda_{k1}\right]\sum_{\ell_{j-1}=0}^{n-j-L^{j}_{j}}\sum_{\ell_{j-2}=0}^{n-j-L^{j}_{j-1}}\ldots\sum_{\ell_{1}=0}^{n-j-L^{j}_{2}}\left(\prod_{k=L_{1}^{j}+j}^{n}\lambda_{k0}-\prod_{k=L_{1}^{j}+j}^{n}\lambda_{k1}\right)^2\\
    &\times\left[\prod_{\mu=1}^{j}(-1)^{\chi_{\mu}\ell_{\mu}}\right]\times\left[\prod_{\mu=1}^{j-1}\left(\prod_{k=L^{j}_{1+\mu}+j-\mu}^{L^{j}_{\mu}+j-\mu}\lambda_{k,(\chi_\mu\chi_k)} - \prod_{k=L^{j}_{1+\mu}+j-\mu}^{L^{j}_{\mu}+j-\mu}\lambda_{k,(1-\chi_\mu\chi_k)}\right)^2\right].
\end{aligned}
\end{equation}
\end{thm}
\begin{proof}
The proof is by double induction, using the base cases $(j,n)=(1,n)$ and $(n,n)$, and we will use the recursive equations to prove the formulas are true for $(j,n+1)$, assuming $(j-1,n)$ and $(j,n)$. Plugging in $j=1$ leaves only the first sum, with $\ell_2=0$. The second product can be ignored, and $L_1^j+j=\ell_1+1$. This reduces exactly to the correct formulas for $j=1$ given in equations \ref{alphaj1} and \ref{betaj1}, similarly so for $j=n$.
\begin{enumerate}
    \item 
If $j=2i$, then the equation we are considering is
\begin{align*}
    \alpha_{j}^{(n+1)}&=\lambda_{10}^2{T}{\alpha}_{j-1}^{(n)}+\lambda_{11}^2\overline{{T}{\alpha}_{j-1}^{(n)}}-2\lambda_{10}\lambda_{11}{T}{\beta}_{j-1}^{(n)}+\lambda_{10}^2{{T}{\alpha}_{j}}^{(n)},\\
    \beta_{j}^{(n+1)}&=\lambda_{10}^2{T}{\alpha}_{j-1}^{(n)}+\lambda_{11}^2\overline{{T}{\alpha}_{j-1}^{(n)}}-2\lambda_{10}\lambda_{11}{T}{\beta}_{j-1}^{(n)}+\lambda_{10}\lambda_{11}{T}{\beta}_{j}^{(n)}.
\end{align*}
It suffices to just check $\alpha_j$ since the difference between the two formulas is just the factors in front.
The first three terms can be written as follows since $j-1$ is odd.
\begin{gather*}
    \lambda_{10}^2{T}{\alpha}_{j-1}^{(n)}+\lambda_{11}^2\overline{{T}{\alpha}_{j-1}^{(n)}}-2\lambda_{10}\lambda_{11}{T}{\beta}_{j-1}^{(n)}
\end{gather*}
\begin{gather*}
    =\sum_{\ell_{j-1}=0}^{n-j+1}\left(\prod_{k=1}^{\ell_{j-1}+1}\lambda_{k\chi_{k+1}}-\prod_{k=1}^{\ell_{j-1}+1}\lambda_{k\chi_{k}}\right)^2\sum_{\ell_{j-2}=0}^{n-j+1-L^{j-1}_{j-1}}\ldots\sum_{\ell_{1}=0}^{n-j+1-L^{j-1}_{2}}\\
    \times\left(\prod_{k=L_{1}^{j-1}+j}^{n+1}\lambda_{k0}-\prod_{k=L_{1}^{j-1}+j}^{n+1}\lambda_{k1}\right)^2\times\left[\prod_{\mu=1}^{j-1}(-1)^{\chi_{\mu}\ell_{\mu}}\right]\\
    \times\left[\prod_{\mu=1}^{j-2}\left(\prod_{k=L^{j-1}_{1+\mu}+j-\mu}^{L^{j-1}_{\mu}+j-\mu}\lambda_{k,(\chi_\mu\chi_k)} - \prod_{k=L^{j-1}_{1+\mu}+j-\mu}^{L^{j-1}_{\mu}+j-\mu}\lambda_{k,(1-\chi_\mu\chi_k)}\right)^2\right].
\end{gather*}
Each of the product indices are raised by 1 due to $T$, and the terms in the last factors of the expression are swapped from this operation, although the expression is the same.
Now we can place the factor in the front into the product at the end, where it takes the place of $\mu=j-1$. We also raise the powers of negative signs by 1 since $j$ is even, and so we are just multiplying by 1. Let $\ell_j=0$ so that we can write the index sum $L_*^{j-1}$ as $L_*^{j}$:
\begin{gather*}
    =\sum_{\ell_{j-1}=0}^{n-j+1}\sum_{\ell_{j-2}=0}^{n-j+1-L^{j}_{j-1}}\ldots\sum_{\ell_{1}=0}^{n-j+1-L^{j}_{2}}\left(\prod_{k=L_{1}^{j}+j}^{n+1}\lambda_{k0}-\prod_{k=L_{1}^{j}+j}^{n+1}\lambda_{k1}\right)^2\\
    \times\left[\prod_{\mu=1}^{j}(-1)^{\chi_{\mu}\ell_{\mu}}\right]\times\left[\prod_{\mu=1}^{j-1}\left(\prod_{k=L^{j}_{1+\mu}+j-\mu}^{L^{j}_{\mu}+j-\mu}\lambda_{k,(\chi_\mu\chi_k)} - \prod_{k=L^{j}_{1+\mu}+j-\mu}^{L^{j}_{\mu}+j-\mu}\lambda_{k,(1-\chi_\mu\chi_k)}\right)^2\right].
\end{gather*}
The last term is $\lambda_{10}^2T\alpha_j^{(n)}$, which can be written as
\begin{equation*}
\begin{aligned}
    &=\sum_{\ell_j=0}^{n-j}\left[\prod_{k=0}^{\ell_j+1}\lambda_{k0}^2\right]\sum_{\ell_{j-1}=0}^{n-j-L^{j}_{j}}\sum_{\ell_{j-2}=0}^{n-j-L^{j}_{j-1}}\ldots\sum_{\ell_{1}=0}^{n-j-L^{j}_{2}}\left(\prod_{k=L_{1}^{j}+j+1}^{n+1}\lambda_{k0}-\prod_{k=L_{1}^{j}+j+1}^{n+1}\lambda_{k1}\right)^2\\
    &\times\left[\prod_{\mu=1}^{j}(-1)^{\chi_{\mu}\ell_{\mu}}\right]\times\left[\prod_{\mu=1}^{j-1}\left(\prod_{k=L^{j}_{1+\mu}+j+1-\mu}^{L^{j}_{\mu}+j+1-\mu}\lambda_{k,(\chi_\mu\chi_k)} - \prod_{k=L^{j}_{1+\mu}+j+1-\mu}^{L^{j}_{\mu}+j+1-\mu}\lambda_{k,(1-\chi_\mu\chi_k)}\right)^2\right]\\
    &=\sum_{\ell_j=1}^{n-j+1}\left[\prod_{k=0}^{\ell_j}\lambda_{k0}^2\right]\sum_{\ell_{j-1}=0}^{n-j+1-L_{j}^{j}}\sum_{\ell_{j-2}=0}^{n-j+1-L^{j}_{j-1}}\ldots\sum_{\ell_{1}=0}^{n-j+1-L^{j}_{2}}\left(\prod_{k=L_{1}^{j}+j}^{n+1}\lambda_{k0}-\prod_{k=L_{1}^{j}+j}^{n+1}\lambda_{k1}\right)^2\\
    &\times\left[\prod_{\mu=1}^{j}(-1)^{\chi_{\mu}\ell_{\mu}}\right]\times\left[\prod_{\mu=1}^{j-1}\left(\prod_{k=L^{j}_{1+\mu}+j-\mu}^{L^{j}_{\mu}+j-\mu}\lambda_{k,(\chi_\mu\chi_k)} - \prod_{k=L^{j}_{1+\mu}+j-\mu}^{L^{j}_{\mu}+j-\mu}\lambda_{k,(1-\chi_\mu\chi_k)}\right)^2\right].
\end{aligned}
\end{equation*}
Adding the terms together yields our desired result by letting $\ell_j=0$ in the formula for $\lambda_{10}^2{T}{\alpha}_{j-1}^{(n)}+\lambda_{11}^2\overline{{T}{\alpha}_{j-1}^{(n)}}-2\lambda_{10}\lambda_{11}{T}{\beta}_{j-1}^{(n)}$. These terms take the $\ell_j$ place in the total sum.
\begin{equation*}
    \begin{aligned}
    \alpha_j^{(n+1)}&=\sum_{\ell_j=0}^{n-j+1}\left[\prod_{k=0}^{\ell_j}\lambda_{k0}^2\right]\sum_{\ell_{j-1}=0}^{n-j+1-L_{j}^{j}}\sum_{\ell_{j-2}=0}^{n-j+1-L^{j}_{j-1}}\ldots\sum_{\ell_{1}=0}^{n-j+1-L^{j}_{2}}\\
    &\times\left[\prod_{\mu=1}^{j}(-1)^{\chi_{\mu}\ell_{\mu}}\right]\times\left(\prod_{k=L_{1}^{j}+j}^{n+1}\lambda_{k0}-\prod_{k=L_{1}^{j}+j}^{n+1}\lambda_{k1}\right)^2\\
    &\times\left[\prod_{\mu=1}^{j-1}\left(\prod_{k=L^{j}_{1+\mu}+j-\mu}^{L^{j}_{\mu}+j-\mu}\lambda_{k,(\chi_\mu\chi_k)} - \prod_{k=L^{j}_{1+\mu}+j-\mu}^{L^{j}_{\mu}+j-\mu}\lambda_{k,(1-\chi_\mu\chi_k)}\right)^2\right].
    \end{aligned}
\end{equation*}

\item
If $j=2i+1$, then the equation we are considering is
\begin{align*}
    \alpha_j^{(n+1)}&=\lambda_{10}^2{T}{\alpha}_{j-1}^{(n)}+\lambda_{11}^2\overline{{T}{\alpha}_{j-1}^{(n)}}-2\lambda_{10}\lambda_{11}{T}{\beta}_{j-1}^{(n)}-\lambda_{11}^2\overline{{T}{\alpha}_{j}^{(n)}},\\
    \beta_j^{(n+1)}&=\lambda_{10}^2{T}{\alpha}_{j-1}^{(n)}+\lambda_{11}^2\overline{{T}{\alpha}_{j-1}^{(n)}}-2\lambda_{10}\lambda_{11}{T}{\beta}_{j-1}^{(n)}-\lambda_{10}\lambda_{11}{T}{\beta}_j^{(n)}.
\end{align*}
It suffices to check $\alpha_j$, for the same reasons as above. In this case $j-1$ is even, so we get
\begin{gather*}
    \lambda_{10}^2{T}{\alpha}_{j-1}^{(n)}+\lambda_{11}^2\overline{{T}{\alpha}_{j-1}^{(n)}}-2\lambda_{10}\lambda_{11}{T}{\beta}_{j-1}^{(n)}
\end{gather*}
\begin{gather*}
    =\sum_{\ell_{j-1}=0}^{n-j+1}\left(\prod_{k=1}^{\ell_{j-1}+1}\lambda_{k0}-\prod_{k=1}^{\ell_{j-1}+1}\lambda_{k1}\right)^2\sum_{\ell_{j-2}=0}^{n-j+1-L^{j-1}_{j-1}}\ldots\sum_{\ell_{1}=0}^{n-j+1-L^{j-1}_{2}}\\
    \times\left(\prod_{k=L_{1}^{j-1}+j}^{n+1}\lambda_{k0}-\prod_{k=L_{1}^{j-1}+j}^{n+1}\lambda_{k1}\right)^2\times\left[\prod_{\mu=1}^{j-1}(-1)^{\chi_{\mu}\ell_{\mu}}\right]\\
    \times\left[\prod_{\mu=1}^{j-2}\left(\prod_{k=L^{j-1}_{1+\mu}+j-\mu}^{L^{j-1}_{\mu}+j-\mu}\lambda_{k,(\chi_\mu\chi_k)} - \prod_{k=L^{j-1}_{1+\mu}+j-\mu}^{L^{j-1}_{\mu}+j-\mu}\lambda_{k,(1-\chi_\mu\chi_k)}\right)^2\right]
\end{gather*}
\begin{gather*}
    =\sum_{\ell_{j-1}=0}^{n-j+1}\sum_{\ell_{j-2}=0}^{n-j+1-L^{j}_{j-1}}\ldots\sum_{\ell_{1}=0}^{n-j+1-L^{j}_{2}}\left(\prod_{k=L_{1}^{j}+j}^{n+1}\lambda_{k0}-\prod_{k=L_{1}^{j}+j}^{n+1}\lambda_{k1}\right)^2\\
    \times\left[\prod_{\mu=1}^{j}(-1)^{\chi_{\mu}\ell_{\mu}}\right]\times\left[\prod_{\mu=1}^{j-1}\left(\prod_{k=L^{j}_{1+\mu}+j-\mu}^{L^{j}_{\mu}+j-\mu}\lambda_{k,(\chi_\mu\chi_k)} - \prod_{k=L^{j}_{1+\mu}+j-\mu}^{L^{j}_{\mu}+j-\mu}\lambda_{k,(1-\chi_\mu\chi_k)}\right)^2\right].
\end{gather*}
As before, we have let $\ell_j=0$ above to make the expression easier to combine with the last term, which can be written as
\begin{gather*}
    -\lambda_{11}^2\overline{{T}{\alpha}_{j}^{(n)}}=-\sum_{\ell_j=0}^{n-j}\left[\prod_{k=0}^{\ell_j+1}\lambda_{k\chi_{k}}^2\right]\sum_{\ell_{j-1}=0}^{n-j-L^{j}_{j}}\sum_{\ell_{j-2}=0}^{n-j-L^{j}_{j-1}}\ldots\sum_{\ell_{1}=0}^{n-j-L^{j}_{2}}\\
    \times\left(\prod_{k=L_{1}^{j}+j+1}^{n+1}\lambda_{k0}-\prod_{k=L_{1}^{j}+j+1}^{n+1}\lambda_{k1}\right)^2\times(-1)^{\ell_{j}}\left[\prod_{\mu=1}^{j-1}(-1)^{\chi_{\mu}\ell_{\mu}}\right]\\
    \times\left[\prod_{\mu=1}^{j-1}\left(\prod_{k=L^{j}_{1+\mu}+j+1-\mu}^{L^{j}_{\mu}+j+1-\mu}\lambda_{k,(\chi_\mu\chi_k)} - \prod_{k=L^{j}_{1+\mu}+j+1-\mu}^{L^{j}_{\mu}+j+1-\mu}\lambda_{k,(1-\chi_\mu\chi_k)}\right)^2\right]
\end{gather*}
\begin{gather*}
    =\sum_{\ell_j=1}^{n-j+1}\left[\prod_{k=0}^{\ell_j}\lambda_{k\chi_k}^2\right]\sum_{\ell_{j-1}=0}^{n-j+1-L_{j}^{j}}\sum_{\ell_{j-2}=0}^{n-j+1-L^{j}_{j-1}}\ldots\sum_{\ell_{1}=0}^{n-j+1-L^{j}_{2}}\\
    \times \left(\prod_{k=L_{1}^{j}+j}^{n+1}\lambda_{k0}-\prod_{k=L_{1}^{j}+j}^{n+1}\lambda_{k1}\right)^2\times(-1)^{\ell_{j}}\left[\prod_{\mu=1}^{j-1}(-1)^{\chi_{\mu}\ell_{\mu}}\right]\\
    \times\left[\prod_{\mu=1}^{j-1}\left(\prod_{k=L^{j}_{1+\mu}+j-\mu}^{L^{j}_{\mu}+j-\mu}\lambda_{k,(\chi_\mu\chi_k)} - \prod_{k=L^{j}_{1+\mu}+j-\mu}^{L^{j}_{\mu}+j-\mu}\lambda_{k,(1-\chi_\mu\chi_k)}\right)^2\right].
\end{gather*}
Adding together yields the formula
\begin{equation*}
\begin{aligned}
    \alpha_j^{(n+1)}&=\sum_{\ell_j=0}^{n-j+1}\left[\prod_{k=0}^{\ell_j}\lambda_{k\chi_k}^2\right]\sum_{\ell_{j-1}=0}^{n-j+1-L_{j}^{j}}\sum_{\ell_{j-2}=0}^{n-j+1-L^{j}_{j-1}}\ldots\sum_{\ell_{1}=0}^{n-j+1-L^{j}_{2}}\\
    &\times \left(\prod_{k=L_{1}^{j}+j}^{n+1}\lambda_{k0}-\prod_{k=L_{1}^{j}+j}^{n+1}\lambda_{k1}\right)^2\times\left[\prod_{\mu=1}^{j}(-1)^{\chi_{\mu}\ell_{\mu}}\right]\\
    &\times\left[\prod_{\mu=1}^{j-1}\left(\prod_{k=L^{j}_{1+\mu}+j-\mu}^{L^{j}_{\mu}+j-\mu}\lambda_{k,(\chi_\mu\chi_k)} - \prod_{k=L^{j}_{1+\mu}+j-\mu}^{L^{j}_{\mu}+j-\mu}\lambda_{k,(1-\chi_\mu\chi_k)}\right)^2\right].
\end{aligned}
\end{equation*}
\end{enumerate}
The formula holds for $j$ even or odd, and so this concludes the proof.
\end{proof}

All that remains is to prove the identities we desire, which seem likely to be true but difficult to show. We make the following conjecture:
\begin{conjecture}\label{conj:values}
For any order $n$, the above formula in equation \ref{betaj_arb} for $\beta_j$ yields the relations $\beta_1=-(-1)^{\frac{n(n+3)}{2}}$,  $\beta_{2i}=(-1)^{\frac{n(n+3)}{2}}\binom{n-i}{i}$, and $\beta_{2i+1}=-(-1)^{\frac{n(n+3)}{2}}\binom{n-i-1}{i}$ for $i\in\mathbb{N}$.
\end{conjecture}
It remains an open question as to whether this conjecture is true, hopefully to be proven in a later paper. However, if this conjecture is proven, it would directly imply Conjecture \ref{conj:convergence} since Theorem \ref{cor:values_vi} shows that the two are equivalent statements.

\subsection{Additional Considerations}
In order for these sequences to actually converge to a single diagonal matrix, we must apply one modification, due to the choice of $\lambda_{jl}$ we made. Recall that the leading term in $\mathcal{A}_k$ is $\alpha_0=\prod\lambda_{k0}^2$. As $b_k\to 0$, the higher order terms vanish, and $\alpha_0$ will only remain as the factor on $a_k$, but this factor means that $a_k$ will rotate on the unit circle endlessly, and if we don't want this, we need to apply an additional matrix to the end of our sequence. Say that our diagonalizing sequence is $\{U_{k+1}\}$ for some order $n$, then multiply by a matrix $F=\diag(\alpha_0^{-1},\alpha_0)=\diag(\prod\lambda_{k1}^2,\prod\lambda_{k0}^2)$. So our sequence is now $\tilde U_{k+1}=U_{k+1}F$, and so our sequences are now multiplied by $\alpha_0^{-1}$. Since this is just a factor with magnitude 1, this still maintains the property $|b_{k+1}|=|b_k|^N$. The benefit to this factor is that now $\mathcal{A}_k$ takes the form $\mathcal{A}_k=1-|b_k|^2(\ldots)\to 1$ as $k\to \infty$. What this means is that the sequence $\{U_{k}^{(n)}\}_{k=0}^{\infty}$ will now converge to the identity gate instead of $a_k$ rotating on the complex unit circle as $k\to\infty$.

\section{Composing Convergent Sequences for Composite Angles}\label{section6}
Recall that we are looking for sequences of the form $U_{k+1}=A_N(U_k;\theta)$, where $A_N$ is some multiplicative function of $U_k$, $U_k^{-1}$, and some diagonal matrices. Here $N\geq3$ is an odd number, $\theta$ is a given angle, and we want to have this sequence to satisfy the conjectured property $|b_{k+1}|=|b_k|^{N}$. If $N$ is composite, then there is a na{\"i}ve way to find a convergent sequence, which is by taking sequence composition. By sequence composition we mean taking two convergent sequences defined by the relations $U_{k+1}=A_{p_1}(U_k;\theta_{p_1})$ and $U_{k+1}=A_{p_2}(U_k;\theta_{p_2})$ respectively and defining a new sequence using the composition $U_{k+1} = A_{p_2}(A_{p_1}(U_k;\theta_{p_1});\theta_{p_2})$. 
This is a very different (and simple) approach to the above, but it is an effective strategy if the only desired outcome is this diagonalization.
\begin{thm}
Provided two sequences $A_{N_1}(U_k;\frac{\pi}{N_1})$ and $A_{N_2}(U_k;\frac{\pi}{N_2})$ satisfying the property from Conjecture \ref{conj:convergence}, we can take the composition of the sequences such that $U_{k+1}=A_{N_1N_2}(U_k;\frac{\pi}{N_1N_2})$ has the property that $|b_{k+1}|=|b_k|^{N_1N_2}$.
\end{thm}
\begin{proof}
Let $V=A_{N_1}(U_k;\pi/N_1)$, with the element $b=V_{21}$ having the property that $|b|=|b_k|^{N_1}$. Then we plug this into the next sequence and get $U_{k+1}=A_{N_2}(V;\pi/N_2)$, and since $A_{N_1}$ preserves the unitary property we have $|b_{k+1}|=|b|^{N_2}=|b_{k}|^{N_1N_2}$. The only part remaining is to manipulate our sequences slightly, since each angle has to be $\theta=\frac{\pi}{N_1N_2}$. Here, let $N=N_1N_2$, and we can write our sequence as
\begin{equation}
    U_{k+1}=A_{N_2}\left(A_{N_1}\left(U_k;\theta\frac{N}{N_1}\right);\theta\frac{N}{N_2}\right).
\end{equation}
Plugging in $\theta=\pi/N$ yields our desired relation.
\end{proof}
By generously applying this theorem, we can construct a prime factorization for a sequence of any composite angle.
\begin{cor}
Let $p_j$ be odd primes such that for all $j=1,\ldots,\ell$, there exists a sequence such that $U_{k+1}=A_{p_j}(U_k;\pi/p_j)$ converges as $|b_{k+1}|=|b_k|^{p_j}$. Then we can compose these sequences together in any order and obtain a sequence $U_{k+1}=A_N(U_k;\theta)$, where $N=p^{m_1}_1\ldots p^{m_\ell}_\ell$ and $\theta=\frac{\pi}{N}$, where $m_j$ are nonnegative integers. This sequence converges as $|b_{k+1}|=|b_k|^N$.
\end{cor}
One way to write this formula is as follows, where $A^{m}_{p}$ is to denote function composition of $A_p$ with itself $m$ times (assume that the angle applied is the same every time):
\begin{equation}
\begin{aligned}
    U_{k+1}^{(N)}&=A_N(U_k;\theta)\\&=A^{m_1}_{p_1}\left(A^{m_2}_{p_2}\bigg(\ldots\left(A^{m_\ell}_{p_\ell}\left(U_k;\theta\frac{N}{p_\ell}\right)\ldots\bigg);\theta\frac{N}{p_2}\right);\theta\frac{N}{p_1}\right)\implies |b_{k+1}|=|b_k|^N.
\end{aligned}
\end{equation}
This corollary also has an interesting consequence: for any composite $N$, we can construct additional sequences alongside the diagonalizing sequence given from the definition that all converge as $|b_{k+1}|=|b_k|^N$, implying that the diagonalizing sequence is not unique up to a phase angle.

Constructing in this manner allows for a simple method of producing sequences for angles of composite $N$. In order to keep the same angle all the way through, you take every instance of $\theta$ in either of the individual sequences and multiply by $N/p_j$, that way you have a modified sequence that can converge for $\theta=\pi/N$ instead of $\pi/p_1$.
This is quite an abstract formulation, so an example helps.
\begin{example}
We can construct a sequence for $N=15$ by using the first two diagonalizing sequences: 
\begin{gather*}
U_{k+1}=U_k D_1(\theta_3) U^{-1}_k D_1(\theta_3) U_k\\
U_{k+1}=U_k D_1(\theta_5) U^{-1}_k D_2(\theta_5) U_k D_2(\theta_5) U^{-1}_k D_1(\theta_5)U_k.
\end{gather*}
Our combined angle is $\theta=\pi/15$, so $\theta_3=5\theta$ and $\theta_5=3\theta$. We can take the composition we defined above, plugging the order 3 sequence into the order 5 sequence:
\[V_k(5\theta) = U_k D_1(5\theta) U^{-1}_k D_1(5\theta) U_k,\quad V_k^{-1}=U_k^{-1} D_1^{-1}(5\theta) U_k D_1^{-1}(5\theta) U_k^{-1}.\]
Combining everything, we get (dropping $(5\theta)$ from $V_k(5\theta)$ for brevity)
\begin{align*}
    \tilde U_{k+1}&=V_k D_1(3\theta) V_k^{-1} D_2(3\theta) V_k D_2(3\theta) V_k^{-1} D_1(3\theta)V_k.
\end{align*}
If we plug in $\theta=\pi/15$, we get exactly what we want: $|b_{k+1}|=|b_k|^{15}$. There is another interesting property of this sequence: it is \textit{not} equivalent to the conjectured diagonalizing sequence for $n=15$ up to a factor, despite having the same convergence rate. The conjectured formulation we have used throughout the paper applies multiples of angles $m\theta$ for $m=1,\ldots,n$, whereas this example consists of applied angles $3\theta$ and $5\theta$. Despite yielding more or less the same end result, the process by which the sequences do so is different. 
\end{example}

\section{Conclusion}\label{section7}
We have provided several results on the nature of these families of diagonalizing sequences, although a few remain to be shown. This paper primarily concerned itself with the critical angles for the sequence for optimal convergence, not the analysis of the behavior in a neighborhood of the angle. Establishing bounds on convergence rates for any $\theta\in (0,\pi/2)$ would be helpful for implementation. Also, these diagonalizing sequences are also being proposed as the optimal algorithm for demonstrating the conjectured convergence rate, but there are likely sequences of this form which converge globally but not optimally. One possible example is given by replacing every diagonal matrix with $D(\theta)=\diag(1,e^{i\theta})$, which appears to still guarantee convergence for any input matrix, just not optimally.

Some mathematical details to the nature of these sequences are needed. In particular, the identity in Conjecture \ref{conj:values} is a very complicated result to show, one that to our knowledge has never been proven. A good mathematical explanation as to why these identities simplify is needed. In addition, the prime factorization argument we provided disproves uniqueness of roots $\lambda_{jl}$ for composite angles, but it is still unknown for prime angles $\pi/p$. If it is unique, then what is special about those roots?

Variants of this scheme can likely be applied for arbitrary states $|s\rangle$ and $|t\rangle$, providing an efficient subroutine for a variety of algorithms with minimal usage of gates.

\noindent \textbf{Acknowledgments.}
S.X.C is partially supported by NSF CCF 2006667, ARO MURI, and  Quantum Science Center (DOE).

\appendix
\section{Trigonometric Identities}\label{app:trig}
Here are all of the trigonometric identities used, as well as a proof since most of them are not commonly used.
\begin{lemma}\label{LemmaA1}
For $\theta=\pi/\left(2n+1\right)$, $\forall n\in\mathbb{N}$,
\begin{equation}
    \sum_{k=0}^{n}(-1)^k\cos(k\theta)=\frac{1}{2}.
\end{equation}
\end{lemma}
\begin{proof}
We can convert this sum into a well-known identity:
$$\sum_{k=0}^{n}(-1)^k\cos(k\theta)=\sum_{k=0}^{n}\cos(k(\theta+\pi))=\frac{\sin((n+1)(\theta+\pi)/2)}{\sin((\theta+\pi)/2)}\cos(n(\theta+\pi)/2).$$
$\theta+\pi=(2n+2)\pi/(2n+1)$, and using the well known trigonometric identity $2\sin \theta \cos \varphi ={\sin(\theta +\varphi )+\sin(\theta -\varphi )}$ we obtain
\begin{align*}
    &=\frac{\sin\left(\frac{(n+1)^2}{2n+1}\pi\right)}{\sin\left(\frac{n+1}{2n+1}\pi\right)}\cos\left(\frac{n^2+n}{2n+1}\pi\right)\\
    &=\frac{1}{2\sin\left(\frac{n+1}{2n+1}\pi\right)}\left[\sin\left(\frac{2n^2+3n+1}{2n+1}\pi\right)+\sin\left(\frac{n+1}{2n+1}\pi\right)\right]\\
    &=\frac{1}{2\sin\left(\frac{n+1}{2n+1}\pi\right)}\left[\sin\left(\frac{(2n+1)(n+1)}{2n+1}\pi\right)+\sin\left(\frac{n+1}{2n+1}\pi\right)\right]=\frac{1}{2}.
\end{align*}
\end{proof}

\begin{lemma}\label{LemmaA2}
For $n\in\mathbb{N}$:
\begin{equation}
    \prod_{k=1}^{n}\sin\left(\frac{(2k-1)\pi}{4n+2}\right)=\frac{1}{2^{n}}.
\end{equation}
\end{lemma}
\begin{proof}
To prove this last step, we transform the product as
\begin{align*}
    \prod_{k=1}^{n}\sin\left(\frac{(2k-1)\pi}{4n+2}\right)=\prod_{k=1}^{n}\cos\left(\frac{\pi}{2}-\frac{(2k-1)\pi}{4n+2}\right)\\
    =\prod_{k=1}^{n}\cos\left(\frac{\left(n+1-k\right)\pi}{2n+1}\right)
    =\prod_{k=1}^{n}\cos\left(\frac{k\pi}{2n+1}\right)\equiv P.
\end{align*}
Multiply by the sine counterpart:
$$Q\equiv\prod_{k=1}^{n}\sin\left(\frac{k\pi}{2n+1}\right),$$
\begin{align*}
    P\cdot Q&=\prod_{k=1}^{n}\cos\left(\frac{k\pi}{2n+1}\right)\sin\left(\frac{k\pi}{2n+1}\right)=\prod_{k=1}^{n}\frac{1}{2}\sin\left(\frac{2k\pi}{2n+1}\right)\\
    &=\frac{1}{2^n}\prod_{k\leq\frac{n}{2}}\sin\left(\frac{2k\pi}{2n+1}\right)\prod_{k>\frac{n}{2}}\sin\left(\frac{2k\pi}{2n+1}\right)\\
    &=\frac{1}{2^n}\prod_{k\leq\frac{n}{2}}\sin\left(\frac{2k\pi}{2n+1}\right)\prod_{k>\frac{n}{2}}\sin\left(\frac{(2n+1-2k)\pi}{2n+1}\right).
\end{align*}
The first product covers every even numbered index, and the second handles all of the odd numbered cases, and so this reduces to
\begin{align*}
    &=\frac{1}{2^n}\prod_{k=1}^n\sin\left(\frac{k\pi}{2n+1}\right)=\frac{1}{2^n} Q\\
    &\implies P\cdot Q=\frac{1}{2^n}Q\implies P=\prod_{k=1}^{n}\cos\left(\frac{k\pi}{2n+1}\right)=\frac{1}{2^n}.
\end{align*}
\end{proof}

\begin{lemma}\label{LemmaA3}
For $\theta=\pi/\left(2n+1\right)$, $\forall n\in\mathbb{N}$,
\begin{equation}
\prod_{k=1}^{n}\left(2\cos(k\theta)+2(-1)^k\right)=\left(-1\right)^{\frac{n\left(n+1\right)}{2}}.
\end{equation}
\end{lemma}
\begin{proof}
This statement is equivalent to saying that for $m\geq 0$, if $n=2m$ or $2m+1$, we have
$$\prod_{k=1}^{2m}\left(2\cos(k\theta)+2(-1)^k\right)=(-1)^m,\quad \prod_{k=1}^{2m+1}\left(2\cos(k\theta)+2(-1)^k\right)=(-1)^{m+1}.$$
For $n=2m$, we can split this product using $2\sin^2(\theta/2)=1-\cos(\theta)$ and $2\cos^2(\theta/2)=1+\cos(\theta)$:
\begin{gather*}
    \prod_{k=1}^{2m}\left(2\cos(k\theta)+2(-1)^k\right)=\prod_{k=1}^{m}\left(2\cos\left(\left(2k-1\right)\theta\right)-2\right)\left(2\cos\left(2k\theta\right)+2\right)\\
    =\prod_{k=1}^{m}-16\sin^2\left(\frac{(2k-1)\theta}{2}\right)\cos^2\left(k\theta\right)=2^{4m}(-1)^m\prod_{k=1}^{m}\sin^2\left(\frac{(2k-1)\theta}{2}\right)\cos^2\left(k\theta\right).
\end{gather*}
It remains to just prove that
\begin{gather*}
    \frac{1}{2^{2m}}=\prod_{k=1}^{m}\sin\left(\frac{(2k-1)\theta}{2}\right)\cos\left(k\theta\right)=\prod_{k=1}^{m}\sin\left(\frac{(2k-1)\pi}{8m+2}\right)\cos\left(\frac{k\pi}{4m+1}\right)\\
    =\prod_{k=1}^{m}\sin\left(\frac{(2k-1)\pi}{8m+2}\right)\sin\left(\frac{(4m+1-2k)\pi}{8m+2}\right).
\end{gather*}
Reversing the product order for the second sine, we get that we can combine the products as 
$$=\prod_{k=1}^{2m}\sin\left(\frac{(2k-1)\pi}{8m+2}\right)=\prod_{k=1}^{n}\sin\left(\frac{(2k-1)\pi}{4n+2}\right)=\frac{1}{2^{n}}=\frac{1}{2^{2m}}.$$
As for the odd case, it is much of the same setup:
\begin{gather*}
    \prod_{k=1}^{2m+1}\left(2\cos(k\theta)+2(-1)^k\right)=\prod_{k=1}^{m+1}\left(2\cos((2k-1)\theta)-2\right)\prod_{k=1}^{m}\left(2\cos(2k\theta)+2\right)\\
    =2^{4m+2}(-1)^{m+1}\prod_{k=1}^{m+1}\sin^2\left(\frac{(2k-1)\theta}{2}\right)\prod_{k=1}^{m}\cos^2\left(k\theta\right).
\end{gather*}
It remains to prove that
\begin{gather*}
    \frac{1}{2^{2m+1}}=\prod_{k=1}^{m+1}\sin\left(\frac{(2k-1)\theta}{2}\right)\prod_{k=1}^{m}\cos\left(k\theta\right)=\prod_{k=1}^{m+1}\sin\left(\frac{(2k-1)\pi}{8m+6}\right)\prod_{k=1}^{m}\cos\left(\frac{k\pi}{4m+3}\right)\\
    =\prod_{k=1}^{m+1}\sin\left(\frac{(2k-1)\pi}{8m+6}\right)\prod_{k=1}^{m}\sin\left(\frac{(4m+3-2k)\pi}{8m+6}\right).
\end{gather*}
Again, reversing the product order in the second product allows us to combine these together:
\[=\prod_{k=1}^{2m+1}\sin\left(\frac{(2k-1)\pi}{8m+6}\right)=\prod_{k=1}^{n}\sin\left(\frac{(2k-1)\pi}{4n+2}\right)=\frac{1}{2^{n}}=\frac{1}{2^{2m+1}}.\]
\end{proof}

\section{Explicit Calculations For n=1,2,3.}
Here is a reference for the explicitly calculated forms for $\alpha_{j}^{(n)}$ and $\beta_j^{(n)}$ for $n=1,2,3$. For order 1 we have the coefficients
\[\alpha_0^{(1)}=\lambda_{10}^2,\quad \beta_0^{(1)}=\lambda_{10}\lambda_{11},\quad \alpha_1^{(1)}=\beta_1^{(1)}=(\lambda_{10}-\lambda_{11})^2.\]
For order 2 we have the coefficients
\[\alpha_0^{(2)}=\lambda_{10}^2\lambda_{20}^2,\quad\beta_0^{(2)}=\lambda_{10}\lambda_{11}\lambda_{20}\lambda_{21},\quad \alpha_2^{(2)}=\beta_2^{(2)}=(\lambda_{10}-\lambda_{11})^2(\lambda_{20}-\lambda_{21})^2,\]
\[\alpha_1^{(2)} = (\lambda_{10}\lambda_{20}-\lambda_{11}\lambda_{21})^2-\lambda_{11}^2(\lambda_{20}-\lambda_{21})^2,\]
\[\beta_1^{(2)} = (\lambda_{10}\lambda_{20}-\lambda_{11}\lambda_{21})^2-\lambda_{10}\lambda_{11}(\lambda_{20}-\lambda_{21})^2.\]
In order 3 we have the coefficients
\[\alpha_0^{(3)}=\lambda_{10}^2\lambda_{20}^2\lambda_{30}^2,\quad\beta_0^{(3)}=\lambda_{10}\lambda_{11}\lambda_{20}\lambda_{21}\lambda_{30}\lambda_{31},\]
\[\alpha_3^{(3)}=\beta_3^{(3)}=(\lambda_{10}-\lambda_{11})^2(\lambda_{20}-\lambda_{21})^2(\lambda_{30}-\lambda_{31})^2,\]
\[\alpha_1^{(3)} = (\lambda_{10}\lambda_{20}\lambda_{30}-\lambda_{11}\lambda_{21}\lambda_{31})^2-\lambda_{10}\lambda_{11}(\lambda_{20}\lambda_{30}-\lambda_{21}\lambda_{31})^2+\lambda_{10}\lambda_{11}\lambda_{20}\lambda_{21}(\lambda_{30}-\lambda_{31})^2,\]
\[\beta_1^{(3)} = (\lambda_{10}\lambda_{20}\lambda_{30}-\lambda_{11}\lambda_{21}\lambda_{31})^2-\lambda_{10}\lambda_{11}(\lambda_{20}\lambda_{30}-\lambda_{21}\lambda_{31})^2+\lambda_{10}\lambda_{11}\lambda_{20}\lambda_{21}(\lambda_{30}-\lambda_{31})^2.\]
\begin{align*}
    \alpha_2^{(3)}=&(\lambda_{11}-\lambda_{10})^2(\lambda_{20}\lambda_{30}-\lambda_{21}\lambda_{31})^2-(\lambda_{11}\lambda_{20}-\lambda_{10}\lambda_{21})^2(\lambda_{30}-\lambda_{31})^2\\
    &+\lambda_{10}^2(\lambda_{20}-\lambda_{21})^2(\lambda_{30}-\lambda_{31})^2,\\
    \beta_2^{(3)}=&(\lambda_{11}-\lambda_{10})^2(\lambda_{20}\lambda_{30}-\lambda_{21}\lambda_{31})^2-(\lambda_{11}\lambda_{20}-\lambda_{10}\lambda_{21})^2(\lambda_{30}-\lambda_{31})^2\\&+\lambda_{10}\lambda_{11}(\lambda_{20}-\lambda_{21})^2(\lambda_{30}-\lambda_{31})^2.
\end{align*}

\bibliographystyle{plain}
\bibliography{main}

\end{document}